\documentclass[a4paper]{article}
\usepackage{amsmath}
\usepackage[english]{babel}
\usepackage{latexsym}
\usepackage{amssymb}
\usepackage{amscd}
\usepackage{amsgen,amstext,amsbsy,amsopn}
\usepackage{math rsfs}	
\usepackage{bm,bbm}
\usepackage{amsthm,epsfig,graphicx,graphics}
\usepackage[latin1]{inputenc}
\usepackage{xspace}
\usepackage{amsxtra}
\usepackage{color}
\usepackage{dsfont}

\usepackage{hyperref}

\usepackage{geometry}
\newcommand{\nt}{\noindent}

\numberwithin{equation}{section}
\newcommand{\bdm}{\begin{displaymath}}
\newcommand{\edm}{\end{displaymath}}
\newcommand{\bdn}{\begin{eqnarray}}
\newcommand{\edn}{\end{eqnarray}}
\newcommand{\bay}{\begin{array}{c}}
\newcommand{\eay}{\end{array}}
\newcommand{\ben}{\begin{enumerate}}
\newcommand{\een}{\end{enumerate}}
\newcommand{\beq}{\begin{equation}}
\newcommand{\eeq}{\end{equation}}
\newcommand{\beqn}{\begin{eqnarray}}
\newcommand{\eeqn}{\end{eqnarray}}
\newcommand{\bml}[1]{\begin{multline} #1 \end{multline}}
\newcommand{\bmln}[1]{\begin{multline*} #1 \end{multline*}}

\newcommand{\lf}{\left}
\newcommand{\ri}{\right}

\newcommand{\braket}[2]{\lf\langle #1|#2 \ri\rangle}

\newcommand{\mean}[3]{\lf\langle \lf.  #1 \lf| #2 \ri| #3 \ri. \ri\rangle}
\newcommand{\RR}{{\mathbb{R}^2}}

\newcommand{\xv}{\mathbf{x}}

\newcommand{\rv}{\mathbf{r}}

\newcommand{\av}{\mathbf{a}}
\newcommand{\deps}{\delta_{\eps}}
\newcommand{\diff}{\mathrm{d}}
\newcommand{\eps}{\varepsilon}

\newcommand{\ab}{\mathcal{A}_{\mathrm{bulk}}}
\newcommand{\ypot}{y_{\mathrm{pot}}}
\newcommand{\hosc}{h_{\mathrm{osc}}}
\newcommand{\kt}{\widetilde{K}}

\newcommand{\gtrial}{g_{\mathrm{trial}}}
\newcommand{\psitrial}{\Psi_{\mathrm{trial}}}

\newcommand{\ba}{\mathcal{B}}

\newcommand{\gpf}{\mathcal{E}^{\mathrm{GP}}}
\newcommand{\gpe}{E^{\mathrm{GP}}}
\newcommand{\gpm}{\psi^{\mathrm{GP}}}
\newcommand{\gpchem}{\mu^{\mathrm{GP}}}

\newcommand{\gpdom}{\mathscr{D}^{\mathrm{GP}}}

\newcommand{\ggpf}{\mathcal{E}^{\mathrm{GP}}_{\mathrm{phys}}}

\newcommand{\tgpf}{\tilde{\mathcal{E}}^{\mathrm{GP}}}

\newcommand{\gvf}{\mathcal{E}^{\mathrm{gv}}}
\newcommand{\gve}{E^{\mathrm{gv}}}

\newcommand{\gvm}{g_{\mathrm{gv}}}
\newcommand{\gvchem}{\mu^{\mathrm{gv}}}
\newcommand{\gvdom}{\mathscr{D}^{\mathrm{gv}}}

\newcommand{\fgv}{F^{\mathrm{gv}}}
\newcommand{\kgv}{K^{\mathrm{gv}}}

\newcommand{\gvfs}{\mathcal{E}_{\star}^{\mathrm{gv}}}
\newcommand{\gves}{E^{\mathrm{gv}}_{\star}}
\newcommand{\gvms}{g_{\star}}

\newcommand{\curl}{\mbox{curl}}

\newcommand{\tfm}{{\rho^{\mathrm{TF}}}}

\newcommand{\tfchem}{\mu^{\mathrm{TF}}}
\newcommand{\xin}{x_{\mathrm{in}}}
\newcommand{\xout}{x_{\mathrm{out}}}

\newcommand{\ann}{{\mathcal{A}_{\eta}}}
\newcommand{\annt}{\widetilde{\mathcal{A}}_{\eta}}
\newcommand{\annh}{\widehat{\mathcal{A}}_{\eta}}
\newcommand{\annd}{\mathcal{A}_{\rm bulk}}
\newcommand{\annm}{\mathcal{A}_{>}}
\newcommand{\anna}{\mathcal{A}_a}
\newcommand{\annmd}{\mathcal{A}_{>}^{\mathrm{2D}}}

\newcommand{\Ofirst}{\Omega_{\mathrm{c_1}}}
\newcommand{\Osec}{\Omega_{\mathrm{c_2}}}
\newcommand{\Othird}{\Omega_{\mathrm{c_3}}}

\newcommand{\disp}{\displaystyle}
\newcommand{\tx}{\textstyle}

\newcommand{\Z}{\mathbb{Z}}
\newcommand{\R}{\mathbb{R}}

\newcommand{\E}{\mathcal{E}}
\newcommand{\A}{\mathcal{A}}

\newcommand{\OO}{\mathcal{O}}

\newcommand{\supp}{\mathrm{supp}}

\newcommand{\Orot}{\Omega_{\mathrm{rot}}}

\newcommand{\Oosc}{\Omega_{\mathrm{osc}}}

\newcommand{\Ophys}{\Omega_{\mathrm{phys}}}
\newcommand{\Omegac}{\Omega_{\mathrm{c}}}

\newcommand{\ep}{\varepsilon}

\newcommand{\Om}{\Omega}

\newcommand{\magnp}{\mathbf{A}_{\mathrm{rot}}}
\newcommand{\half}{\tx\frac{1}{2}}

\newcommand{\aavoo}{\mathbf{A}_{\Omega}}

\newcommand{\rmax}{R_{\mathrm{m}}}

\usepackage{cancel}

\newcommand{\bigO}[1]{\mathcal{O}\lf(#1\ri)}

\newcommand{\inpr}[3]{\lf\langle#1\lf|#2\ri|#3\ri\rangle_\eta}

\newcommand{\xiny}{1+\eps^2y}
\newcommand{\xint}{1+\eps^2t}
\newcommand{\xinyp}{\lf(1+\eps^2y\ri)}
\newcommand{\xintp}{\lf(1+\eps^2t\ri)}
\newcommand{\dey}{\diff y \: \xinyp \:}

\newcommand{\yb}{{y_\beta}}

\newcommand{\gvbf}{\mathcal{E}_\beta^\mathrm{gv}}
\newcommand{\gvbe}{E_\beta^\mathrm{gv}}
\newcommand{\gvbm}{g_\beta}
\newcommand{\gvbchem}{\mu_\beta}
\newcommand{\gvbdom}{\mathscr{D}_\beta^\mathrm{gv}}

\newcommand{\gcre}{E_{\star}^\mathrm{gv}}
\newcommand{\gcrm}{g_{\star}}
\newcommand{\gcrchem}{\mu_{\star}}
\newcommand{\betas}{{\beta_{\star}}}

\theoremstyle{remark}

\theoremstyle{plain}

\newtheorem{teo}{Theorem}[section]
\newtheorem{lem}{Lemma}[section]
\newtheorem{pro}{Proposition}[section]
\newtheorem{cor}{Corollary}[section]

\newtheorem{rem}{Remark}[section]

\begin{document}

\markboth{\scriptsize{\textsc{Correggi, Dimonte} -- Third Critical Speed for Rotating BECs}}{\scriptsize{\textsc{Correggi, Dimonte} -- Third Critical Speed for Rotating BECs}}

\title{On the Third Critical Speed for Rotating Bose-Einstein Condensates}

\author{M. Correggi${}^{a}$, D. Dimonte${}^{b}$
	\\
	\normalsize\it ${}^{a}$ Dipartimento di Matematica e Fisica, Universit\`{a} degli Studi Roma Tre,	\\
	\normalsize\it L.go San Leonardo Murialdo, 1, 00146, Rome, Italy.	\\
	\normalsize\it ${}^{b}$  Scuola Internazionale Superiore di Studi Avanzati,		\\
	\normalsize\it Via Bonomea, 265, 34136 Trieste, Italy.}
	
\date{\today}

\maketitle

\begin{abstract} 
We study a two-dimensional rotating Bose-Einstein condensate confined by an anharmonic trap in the framework of the Gross-Pitaevksii theory. We consider a rapid rotation regime close to the transition to a giant vortex state. It was proven in \cite{CPRY3} that such a transition occurs when the angular velocity is of order $ \eps^{-4}$, with $ \eps^{-2} $ denoting the coefficient of the nonlinear term in the Gross-Pitaevskii functional and $ \eps \ll 1 $ (Thomas-Fermi regime). In this paper we identify a finite value $ \Omegac $ such that, if $ \Omega = \Omega_0/\eps^4 $ with $ \Omega_0 > \Omegac $, the condensate is in the giant vortex phase. Under the same condition we prove a refined energy asymptotics and an estimate of the winding number of any Gross-Pitaevskii minimizer.
\end{abstract}

\tableofcontents

\section{Introduction}

\nt
Since the first experimental realization of {\it Bose-Einstein (BE) condensation} in the 90's, BE condensates and cold atoms in general have been extensively studied to investigate quantum properties on almost macroscopic scales. Among the typical features of BE condensates, one of the most striking is certainly  {\it superfluidity}, which has been studied in several experiments in the last years by putting the quantum system under rotation and observing its response (see, e.g., the reviews \cite{Co,Fe1}). Because of the quantum nature of BE condensates, the only possible change to the condensate profile due to the imposed rotation is the nucleation of isolated defects, i.e., quantum {\it vortices}. The generation of vortices has been observed in various experiments as well as the growth of their number when the angular velocity increases \cite{MCWD,RAVXK,CHES}. For even larger angular velocities the number of vortices becomes so large that they fill the bulk of the system and arrange in a typical Abrikosov lattice \cite{ARVK}. In presence of harmonic trapping the rotation can not be arbitrarily fast, otherwise the centrifugal forces would break down the trapping and system would eventually fly apart. On the opposite when the trapping contains some stronger confinement, e.g., some anharmonic potential growing faster that $ |\rv|^2$ for large $ |\rv| $, regimes with much more rapid rotation can in principle be reached. Unfortunately so far a loss of coherence of the system has prevented the exploration of such regimes in the experiments \cite{BSSD}, although a depression at the center of the trap has been observed for large angular velocities.

However it has been predicted \cite{CD,CDY2,Fe2,FJS,FB,KTU,KB,KF,R1} that, besides the nucleation of vortices, other phase transitions should be observed in rapid rotating condensates in case of anharmonic confinement, with the occurrence of macroscopic defects or the transition to {\it giant vortex} states: when the rotational velocity gets very large, the centrifugal forces constrain the condensate in some thin annular region around a macroscopic hole and, if the rotation gets even more rapid, vortices disappear from the bulk of the system, which seems then to carry a  huge circulation centered at the origin.

Although BE condensates are many-body quantum system composed of a number of atoms ranging from few thousands to many millions, all the physical prediction about them are made by using an effective theory, the {\it Gross-Pitaevskii (GP) theory}, namely a one-particle approximation in which the energy of the system is given by a suitable nonlinear functional (see below). In spite of its simplicity the agreement with experimental observations is quite good, specially in the so called Thomas-Fermi regime, i.e., when the effective coupling becomes large. One of the major advantages of GP theory is the possibility of run very sophisticated and accurate numerical simulations \cite{Dan,FJS,KTU}. See also the webpage \texttt{http://gpelab.math.cnrs.fr/}, where one can find an efficient free code for simulations of the GP energy or dynamics developed by \textsc{X. Antoine} and \textsc{R. Duboscq} \cite{AD1,AD2}.

In the framework of the GP theory the energy of a two-dimensional rotating BE condensate in physical units on the plane orthogonal to the rotational axis reads
\beq
\label{eq: gGPf}
	\ggpf[\Psi] = \int_{\mathbb R^2} {\diff} \rv \: \lf\{ \half\lf|\lf(\nabla - i \magnp\ri) \Psi \ri|^2 + 
\lf( V(r) - \half \Omega_{\rm rot}^2 r^2 \ri) |\Psi|^2 + \frac{|\Psi|^4}{\eps^2} \ri\},
\eeq
where $ \Orot $ is the angular velocity, $ \magnp : = \Orot r \mathbf{e}_{\vartheta} $,  $r=|\mathbf r|$, with $\mathbf r = (x,y) \in\mathbb R^2$, and $ \mathbf{e}_{\vartheta} = (-y,x)/|\rv| $ is the unit vector in the transverse direction. The trapping potential is assumed to be of the form  
\beq
	\label{eq: potential}
	V(r) : = k  r^s + \half \Oosc^2 r^2
\eeq
with $k>0$ and
\beq 
	2 < s <\infty,
\eeq
i.e., the harmonic trapping is corrected by some anharmonic perturbation. Finally, we will focus on the study of the Thomas-Fermi (TF) regime $\ep \to 0$. The ground state energy of the system is thus obtained by minimizing the functional \eqref{eq: gGPf} under the normalization constraint $ \lf\| \Psi \ri\|_2 = 1 $, which amounts to require conservation of the particle number. Any minimizer is called condensate wave function and its modulus square, i.e., the associated probability distribution, is what can be observed experimentally.

The range of validity of the GP description as well the derivation of the GP effective theory from the quantum mechanical description of a condensed Bose gas is an interesting topic on its own, which has been completely solved in the non-rotating case \cite{LSY1,LSSY}. In presence of rotation on the other hand \cite{LS} contains a derivation of the GP functional, which is however restricted to bounded angular velocities and therefore not directly applicable to the case under discussion (see also \cite{BCPY} for further results). However we will not investigate further such questions and take as a starting point the GP theory.

The mathematical physics literature contains now a large number of works studying the behavior of the GP minimization problem in different asymptotic regimes of the angular velocity. If we restrict the discussion to trapping potentials of the form \eqref{eq: potential}, three {\it phase transitions} have been identified (see \cite{CPRY3} for an extensive discussion or \cite{CPRY4} for a more concise exposition), corresponding to three critical values of the rotational velocity. Here we briefly sum up the most relevant features of the physics of rotating condensates in anharmonic traps:
\begin{itemize}
	\item for small angular velocities $ \Orot $, the rotation has no effect on the condensate wave function, i.e., the minimizer of $ \ggpf $ coincides with the one in absence of rotation \cite{AJR};
	\item when the first critical speed $ \Ofirst \propto \eps^{\frac{4}{s+2}}|\log\eps| $ is crossed, one observes the nucleation of quantum vortices, i.e., isolated zero of the condensate wave function \cite{CR1};
	\item if $ \Orot $ stays far from a second critical speed $ \Osec \propto \eps^{-\frac{s-2}{s+2}} $, the number of vortices might increase but the profile of the condensate wave function is still close to the non-rotating one. Close to $ \Ofirst $ it is possible to derive the explicit distribution of vortices \cite{CR1}, which eventually cover the whole bulk of the condensate. In this regime one expects that they arrange in a regular (Abrikosov) lattice to minimize the interaction energy. This remains an open question although it has been proven that the vorticity is uniformly distributed \cite{CPRY3};
	\item for $ \Orot \propto \eps^{-\frac{s-2}{s+2}} $ a first change of the macroscopic profile of the condensate is observed, due to the effect of the centrifugal forces. When the second critical speed $ \Osec $ is crossed this change has a dramatic effect since a macroscopic hole is created at the center of the trap. However the vorticity remains uniform in the bulk of the system \cite{CPRY3};
	\item for very rapid rotations above $ \Osec $, the bulk of the condensate becomes essentially annular and its width shrinks as $  \eps \to 0 $. No further changes are however observed until a third critical speed $ \Othird \propto \eps^{-\frac{4(s-2)}{s+2}} $ is crossed. Then vortices are expelled from the bulk and the condensate behaves as if the whole vorticity was concentrated at the origin of the trap \cite{CPRY3,CPRY5}. This is the giant vortex state that we plan to study in this paper.
\end{itemize}

So far we have only discussed condensates in anharmonic traps of the type \eqref{eq: potential} but a lot of results are also available for other classes of trapping potentials. First of all the harmonic case has been extensively studied both in the physics and mathematical literature and, while there exists a first critical value of the angular velocity \cite{IM1,IM2} corresponding to the occurrence of vortices and the behavior of the condensate for not too rapid rotation is similar to the one described above (vortex lattice, uniform distribution of vorticity, etc.), when the angular velocity approaches the harmonic frequency of the trap, some new physical features come into play and fractional quantum Hall states emerge \cite{ABD,ABN,LSY2}. As we have already mentioned larger angular velocity are not allowed because the system would otherwise be no longer trapped. See however \cite{Ka} for an alternative setting in which the trapping potential is suitably rescaled to reach fast rotation regimes.

Even if we restrict to the anharmonic traps \eqref{eq: potential} there is an extreme case which is of certain interest, namely $ s = \infty $. Formally this corresponds to a confinement of the system to a two-dimensional disc of unit radius. Naturally one has then to provide suitable boundary conditions and both the Neumann \cite{CDY1,CY,CRY} and Dirichlet \cite{CPRY1} cases have been deeply studied. Indeed phase transitions analogous to the one described above has been found out even in this extreme case, although the nature of the third one is much more subtle. 

Let us now go back to the functional \eqref{eq: gGPf} and introduce more convenient parameters: if we set 
\beq
	\Ophys : = \sqrt{\Orot^2 - \Oosc^2},
\eeq
and obviously assume that $ \Oosc < \Orot $, the trapping potential can be cast in the form
\beq
	\label{eq: eff potential}
	V(r) = k r^s - \tx\frac{1}{2} \Ophys^2 r^2.
\eeq
Since we are interested in exploring a regime in which both $ \eps \to 0 $ and $ \Ophys \to \infty $ (or, equivalently,  $ \Orot \to \infty $), it is convenient to rescale units in the GP functional, in order to observe a non-trivial behavior \cite[Sect. I.A]{CPRY3}: if one would naively minimize $ \ggpf $ under the mass constraint $ \| \Psi \| = 1$, one would get trivially that the ground state energy diverges and the corresponding minimizer tends to $ 0 $ pointwise. The appropriate rescaling depends however on the asymptotics of $  \Ophys $ and, in the regime we want to explore (very fast rotation), 
\beq
	\label{eq: lower bound speed}
	\Ophys \gg \eps^{-\frac{s-2}{s+2}},
\eeq
which leads to the rescaling (see \cite[Sect. I.A]{CPRY3})
\beq
 	\label{eq: rescaling}
	\rv =  \rmax \xv,	\qquad	\Psi(\rv) = \rmax^{-1} \psi(\xv),	\qquad \Ophys =  \rmax^{-2} \Omega,	\qquad	\aavoo  = \Omega x \mathbf{e}_{\vartheta},
\eeq
where $ \rmax $ stands for the unique minimum point of the potential \eqref{eq: eff potential}, i.e., explicitly
\beq
	\label{eq: rm}
	\rmax : = \lf( \frac{\Ophys^2}{s k} \ri)^{\frac{1}{s-2}}.
\eeq
Under the scaling \eqref{eq: rescaling}, the GP functional \eqref{eq: gGPf} becomes
\beq	
	\ggpf[\Psi] = \rmax^{-2} \lf[ \gpf[\psi]  + \lf( \tx\frac{s}{2} - \half \ri) \Omega^2 \ri],
\eeq
with
\beq
	\label{eq: gpf}
	\framebox{$\gpf[\psi] : = \disp\int_{\R^2} \diff \xv \lf\{ \half \lf| \lf( \nabla - i \aavoo \ri) \psi \ri|^2 + \Omega^2 W(x) |\psi|^2 + \eps^{-2} |\psi|^{4} \ri\}$}.
\eeq
The rescaled potential 
\beq
	\label{eq: W}
	W(x) : = \tx\frac{1}{s} \lf( x^s - 1 \ri) - \frac{1}{2} \lf( x^2 - 1 \ri)
\eeq
is positive and has a unique minimum at $ x = 1 $, i.e., $ \inf_{x \in \R^+} W(x) = W(1) = 0 $. The rescaled angular velocity $ \Omega $ is related to the original physical quantities via
\beq
	 \Omega = (sk)^{-\frac{2}{s-2}} \Omega_{\rm phys}^{\frac{s+2}{s-2}},
\eeq
and condition \eqref{eq: lower bound speed} becomes
\beq
	\Omega \gg \eps^{-1}.
\eeq 

From now on we will focus on the analysis of the minimization of the functional \eqref{eq: gpf} on the domain
\beq
	\label{eq: domain}
	{\gpdom : = \lf\{ \psi \in H^1(\R^2)\: \big| x^{s/2} \psi \in L^2(\R^2), \lf\| \psi \ri\|_2 = 1 \ri\}}.
\eeq
We also set
\beq
	\label{eq: ground state}
	\gpe : = \inf_{\psi \in \gpdom} \gpf[\psi],	
\eeq
and denote by $ \gpm $ any minimizer, which is known to exist by standard arguments. In addition any $ \gpm $, which might be non-unique due to a breaking of the rotational symmetry and the occurrence of isolated vortices, solves the variational equation
\beq
	\label{eq: GP variational}
	- \half \lf(\nabla - i \aavoo \ri)^2 \gpm  + \Omega^2 W(x) \gpm + 2 \eps^{-2} \lf| \gpm \ri|^2 \gpm = \gpchem \gpm,
\eeq
where the chemical potential (Lagrange multiplier) is fixed by imposing the $L^2-$normalization of $ \gpm $:
\beq
	\label{eq: chem}
	\gpchem = \gpe + \eps^{-2} \int_{\R^2} \diff \xv \: \lf| \gpm \ri|^4.
\eeq

As discussed in details in \cite[Sect. I.B]{CPRY3}, when $ \Omega \gg \eps^{-1} $ the condensate has already crossed the second critical speed, i.e., its profile approaches a density supported on an annulus centered in the origin, whose inner and outer radii tend to $ 1 $ as $ \eps \to 0 $. More precisely $ |\gpm|^2 $ is close {in $ L^p $, $ p < \infty $,} to the TF profile
\beq
	\label{eq: tfm}
	\tfm(x) = \tx\frac{1}{2} \lf[ \tfchem - \eps^2 \Omega^2 W(x) \ri]_+,
\eeq
with $  \tfchem $ the chemical potential fixed by the $ L^1$-normalization of the function. A straightforward analysis shows indeed that $ \tfm $ is compactly supported and $ \supp(\tfm) = \lf[ \xin, \xout \ri] $ with \cite[Eq. (2.7)]{CPRY3}
\beq
	\xout - \xin = C \lf(\eps \Omega\ri)^{-2/3} {\ll 1},		\qquad		{x_{\mathrm{in}/\mathrm{out}} = 1 + \OO\big(\lf(\eps \Omega\ri)^{-2/3}\big)},
\eeq
{as it can be proven by taking a Taylor expansion of $ W $ around $ x = 1 $ in \eqref{eq: tfm} and imposing the $ L^1 $ normalization.}

The vortex structure of $ \gpm $ is richer: being well above the first critical speed $ \Ofirst \sim |\log\eps| $ for the nucleation of vortices, the GP minimizer contains a very large number of vortices distributed all over its support. More precisely one can prove that the vorticity is uniformly distributed in the bulk of the condensate. {As in \cite[Eq. (1.42)]{CPRY3}, we denote by $ {\mathcal R}_{\rm bulk} \subset \supp(\tfm) $ a suitable annulus $ \lf\{ \xv \: | \: x_< \leq x \leq x_> \ri\} $ with $ x_{>/<} = x_{\mathrm{out}/\mathrm{in}} + o((\eps \Omega)^{-2/3}) $}.
	
	\begin{teo}[\textbf{\cite[Theorem 1.2]{CPRY3}}]
		\label{teo: vortex distribution CPRY3}
		\mbox{}	\\
		If $ \eps^{-1} \lesssim \Omega \ll \eps^{-4} $ as $ \eps \to 0 $, there exists a finite family of disjoint balls $ \{ \ba_i \} : = \{ \ba(\xv_i, \varrho_i) \} \subset {\mathcal R}_{\rm bulk} $, $ i = 1, \ldots, N $, such that 		\ben
			\item 
$ \varrho_i \leq \OO\left(\Omega^{-1/2}\right) $, 
$ \sum \varrho_i^2 \leq (1 + (\eps\Omega)^{2/3})^{-1} $;
			\item $ \lf|\gpm\ri| > 0 $ on $ \partial \ba_i $, $ i = 1, \ldots, N $.
		\een
		Moreover, setting $ d_i : = \deg\{\gpm, \partial \ba_i\} $ and defining the vorticity measure as  $ \mu : =  2\pi\sum_{i = 1}^N d_i \delta(\xv- \xv_i) $, then, for any set $ \mathcal{S} \subset {\mathcal R}_{\rm bulk} $ such that {$ |\partial \mathcal{S}| = 0 $ and} $ |\mathcal{S}| \gg \Omega^{-1} |\log(\eps^4\Omega)|^2 $ as $ \eps \to 0 $,  
		\beq
			\frac{\mu(\mathcal{S})}{\Omega |\mathcal{S}|} \underset{\eps \to 0}{\longrightarrow} 1.
		\eeq
	\end{teo}
	
	\nt
	The inner region $ \lf\{ \xv \in \R^2 \: | \: x \leq \xin \ri\} $ is presumably also filled with vortices but, because of the exponential smallness of $ \gpm $ there, the vortex structure in that region is practically inaccessible. An important condition contained in Theorem \ref{teo: vortex distribution CPRY3} is the request 
	\bdm
		\Omega \ll \eps^{-4}.
	\edm
	The reason is that at angular velocities of that order the proof of Theorem \ref{teo: vortex distribution CPRY3} might fail due to the occurrence of a further phase transition, i.e., the transition to a giant vortex state. This paper is precisely devoted to the investigation of such a transition. 
	
	From the heuristic point of view it is quite simple to explain why one should expect a change in the vortex structure when $ \Omega \sim \eps^{-4} $: from energy considerations it is easy to see that the average size of the vortex core, i.e., the radius of the region around a vortex when $ |\gpm|^2 $ is substantially far from $ \tfm $, is of order $ \eps^{2/3} \Omega^{-1/3} $. The width of $ \supp(\tfm) $ is on the other hand of order $ (\eps\Omega)^{-2/3} $ and the two quantities are clearly of the same order when $ \Omega \sim \eps^{-4} $. Hence it must happen that for $ \Omega = \Omega_0 \eps^{-4} $ with $ \Omega_0 $ a large enough constant, the vortex core becomes larger that the bulk of the condensate, i.e., vortices can no longer be accommodated in $ \supp(\tfm) $. A non trivial phase factor of $ \gpm $ is however needed in order to compensate the effect of the rotation but, because no vortex can occur in the bulk of the condensate, all the vorticity should get concentrated in the inner region where $ \gpm $ is exponentially small. In fact when this occurs it is impossible to distinguish from the energetic point of view such a state with vortices distributed in the inner hole from a giant vortex state of the form $ f(x) e^{i n \vartheta} $, $ n \in \Z $. 
	
	Notice that although this might seem to suggest that the rotational symmetry is restored, such a phenomenon never occurs as proven in \cite[Theorem 1.6]{CPRY3}. However the GP energy is expected to be well approximated above the critical speed for the transition to a giant vortex state by a one-dimensional energy functional obtained by evaluating $ \gpf $ on functions of the form $ f(x) e^{i n \vartheta} $. In fact by some very simple observations one can show that $ n = \lfloor \Omega \rfloor (1 + o(1)) $, where $ \lfloor \: \cdot \: \rfloor $ stands for the integer part. Let us now fix the angular velocity to be
	\beq
		\label{eq: speed}
		\Omega = \frac{\Omega_0}{\eps^4},
	\eeq
	with $ \Omega_0 $ a positive constant. Concerning the giant vortex regime, the main results proven in \cite{CPRY3} are stated below. {We denote by $ \ab $ a suitable annular layer around $ x = 1 $ containing the bulk of the condensate (see next \eqref{eq: bulk} for a precise definition).}
	
	\begin{teo}[\textbf{\cite[Theorem 1.3]{CPRY3}}]
		\label{teo: giant vortex CPRY3}
		\mbox{}	\\
		If $ \Omega $ is given by \eqref{eq: speed}, there exists a finite constant $ \bar\Omega_0 $ such that for any $ \Omega_0 > \bar\Omega_0 $, no minimizer $ \gpm $ has a zero inside $ \ab $ if $ \eps $ is sufficiently small.
	\end{teo}
	
	\begin{teo}[\textbf{\cite[Theorem 1.4]{CPRY3}}]
		\label{teo: gv energy CPRY3}	
		\mbox{}	\\
		If $ \Omega $ is given by \eqref{eq: speed} with $ \Omega_0 > \bar\Omega_0 $  as in Theorem \ref{teo: giant vortex CPRY3}, then as $ \eps \to 0 $\footnote{{We use here polar coordinates $ (x,\vartheta) \in \R^+ \times [0,\pi) $ on the plane.}}
		\beq
			\label{teo: gpe asympt CPRY3}
			\gpe = \min_{\lf\| f \ri\|_2 = 1} \gpf\lf[f(x) e^{i \lfloor \Omega \rfloor \vartheta}\ri] + \OO(|\log \ep| ^{9/2}).
		\eeq
	\end{teo}
		
	\nt 
	The first result, although being a consequence of the energy asymptotics \eqref{teo: gpe asympt CPRY3}, is the most relevant one, since it shows the occurrence of the giant vortex transition for angular velocities of order $ \eps^{-4} $. The precise mathematical statement is a pointwise estimate in the bulk region of $ |\gpm| $ in terms of a strictly positive function, i.e., the minimizer of the functional appearing on the r.h.s. of \eqref{teo: gpe asympt CPRY3}: since the latter is bounded from below by a positive constant in the bulk and the difference is pointwise small in $ \eps $, also $ \gpm $ can not vanish there. 
	
	For the analysis of the present paper it is very important to remark that both results hold true if the angular velocity is expressed by \eqref{eq: speed} with $ \Omega_0 $ {\it large enough}, namely no precise estimate is {derived there} on the sharp value for the transition {(see also Remark \ref{rem: comparison})}. We indeed expect that the giant vortex structure appears as soon as $ \Omega $ becomes (asymptotically) larger than
	\beq
		\label{eq: othird}
		\Othird = \frac{\Omegac}{\eps^4},
	\eeq
	for some explicit value $ \Omegac $. In this paper we will indeed investigate such a question and exhibit a finite value $ \Omegac $ which is a good candidate for the sharp constant. Actually we are going to see that such a constant is a solution of some algebraic equation (see \eqref{eq: Omegac}) involving quantities relative to a limit problem independent of $ \eps $. Although we have not proven it yet, we do expect that next \eqref{eq: Omegac} has a unique solution, thus providing the sharp value of the critical velocity\footnote{Strictly speaking in order to show that $ \Omegac $ is the sharp value for the transition one should also prove that, below $ \Omegac $, vortices are still present in the bulk of the condensate, as done in \cite{R2} for hard anharmonic traps. We will come back to this question later.}. 
	
	We outline here the structure of the paper. Next Section contains the main results, i.e., the identification of the explicit value of the angular velocity for the transition to the giant vortex state, together with an asymptotic expansion of the GP ground state energy which is actually on the main ingredients of the proof of the above mentioned result. We also show that, as in \cite[Theorem 1.5]{CPRY3}, one can deduce a (better) estimate of the total winding number of any GP minimizer.
	
	Sections \ref{sec: preliminary} contains some preliminary estimates and a detailed analysis of the effective functionals that will play a significant role throughout the proofs. In Section \ref{sec: positivity} we prove the main properties of the cost function and in particular its positivity, which is the main mathematical tool used in the proof of the giant vortex transition as in several other works \cite{CPRY1, CPRY2, CPRY3, CR1, CRY}.
	
	Sections \ref{sec: energy} and \ref{sec: gv transition} are devoted to the proofs of the main results: we first (Section \ref{sec: energy}) obtain the asymptotic expansion of the GP energy by comparing suitable upper and lower bounds and then (Section \ref{sec: gv transition}) use such a result to deduce the pointwise estimate of $ |\gpm| $ showing the absence of vortices in the bulk. 

\medskip
\nt
{{\bf Notation:} In the asymptotic analysis $ \eps \to 0  $ we will often use the Landau symbols: given a positive function $ g $, we say that 
	\begin{itemize}
		\item 	$ f = \mathcal{O}(g) $ (resp. $ = o(g) $), if $ \lim_{\eps \to 0} |f|/g \leq C < \infty $ (resp. $ = 0 $);
		\item 	$ f \propto g $, whenever $ \lim_{\eps \to 0} |f|/g = C $, with $ 0 < C < \infty $;
		\item 	if $ f \geq 0 $, $ f \ll g $ is synonimous of $ f = o(g) $ and $ f \gg g $ simply means that $ g \ll f $.
	\end{itemize}
	Sometimes we will use the notation $ \OO(|\log\eps|^{\infty}) $ to indicate a quantity of order $ |\log\eps|^a $ for some finite but possibly large $ a $. Since such a quantity will typically appear multiplied by powers of $ \eps $, the explicit value of $ a $ will be irrelevant. 	\\
	We denote by $ \ba_{\varrho}(\xv) $ any two-dimensional ball centered in $ \xv $ and with radius $ \varrho $ and by $ \lfloor x \rfloor $ the integer part of the real number $ x $. The symbol $ C $ will stand for a finite constant independent of $ \eps $, whose value might change from line to line.
}

\medskip
\nt
{\bf Acknowledgements:} The authors acknowledge the support of MIUR through the FIR grant 2013 ``Condensed Matter in Mathematical Physics (COND-MATH)'' (code RBFR13WAET).

\section{Main Results}
\label{sec: main results}

The first non trivial observation to improve the results proven in \cite{CPRY3} is that instead of making a special choice of the giant vortex winding number ($ \lfloor \Omega \rfloor $ in \cite{CPRY3}), one might try and optimize w.r.t. such a parameter, so obtaining a better candidate for the giant vortex state. This leads to consider the functional obtained evaluating $ \gpf $ on a giant vortex ansatz $ f(x) e^{i n \vartheta}$ and minimize w.r.t. both $ f $ and $ n $ to find out the optimal giant vortex phase, i.e., explicitly
\beq
	\label{eq: gvbf}
	\gvbf[g] = \int _{-\eta}^{\eta} \diff y \: (1 + \eps^2 y) \lf\{ \half
		\lf|
			\nabla g(y)
		\ri| ^2
		+U _\beta (y) g ^2 (y)
		+\varepsilon ^2 y ^3 v(y) g ^2(y)
		+\frac{1}{2\pi}
		g ^4(y)
	\ri\},
\eeq
where we have set for convenience $ n = \Omega + \beta $ and exploited the exponential fall off of $ \gpm $ to cut the tails $ |y| \geq \eta \propto |\log\eps| $. The spatial coordinate has also been rescaled around $ |\xv| = 1 $ by setting $ x = 1 + \eps^2 y $. The potentials $ U_\beta $ and $ v $ are obtained via a Taylor expansion of $ W(x) $ around $ x = 1 $ and to the leading order in $ \eps $ are simply given by a shifted quadratic potential (see \eqref{eq: pot U} and \eqref{eq: pot v} for their explicit expressions). We remark however that in $ U_\beta $, the parameter $ \beta $ always appears multiplied (at least) by $ \eps^2 $, so showing that the correction is only lower order. 

Setting $ \gvbe : = \inf_{\| f \| = 1} \gvbf[f] $ and denoting by $ \gvbm $ the corresponding minimizer, which can be proven to exist and be unique (up to multiplication by a phase factor) (see Proposition \ref{pro: min gvbf}), one can subsequently minimize w.r.t. $ \beta \in \R $, obtaining the energy $ \gcre $, an optimal phase $ \betas $ and a density $ \gcrm $, i.e.,
\beq
	\label{eq: gcre}
	\gcre := \min_{\beta \in \R} \gvbe = \gve_{\betas} = \gvf_{\betas}[\gcrm].
\eeq

In Subsection \ref{sec: betas} we will prove that $ \betas = \OO(1) $, so that, by the above argument, one expects the functional $ \gvbf $ to be close in the limit $ \eps \to 0 $ to the following simplified giant vortex functional
\beq
	\label{eq: gvf}
	\gvf[g] = \int _{\R} \diff y
	\lf\{
		\tx\frac{1}{2}\lf(
			g '
		\ri) ^2
		+ \frac{\alpha ^2}{2} y^2 g ^2 
		+\frac{1}{2\pi} g ^4
	\ri\},
\eeq 
with ground state energy $ \gve $ and minimizer $ \gvm $, i.e.,
\beq
	\gve : = \inf_{g \in \gvdom} \gvf[g] = \gvf[\gvm],
\eeq
where 
\bdm
	\gvdom : = \lf\{ g \in H^1(\R) \: \big| \: y g \in L^2(\R), \lf\| g \ri\|_{L^2(\R)}  = 1 \ri\}.
\edm
Here we have denoted for short
\beq
	\label{eq: alpha}
	\alpha : = \Omega_0 \sqrt{s+2}.
\eeq
The minimizer $ \gvm $ solves the variational equation
\beq
	- \tx\frac{1}{2} g^{\prime\prime} + \frac{1}{2} \alpha^2 y^2 g + \frac{1}{\pi} g^3 = \gvchem g,
\eeq
where $ \gvchem = \gve + \frac{1}{2\pi} \lf\| \gvm \ri\|_4^4 $.

We are now in position to introduce the explicit value of the constant $ \Omegac $ appearing in the critical value of the angular velocity $ \Othird $, which can be expressed in terms of the critical quantities associated with the effective one-dimensional functional $ \gvf $ and, specifically, $ \gvm $ and $ \gvchem $: we denote by $ \Omegac $ the {\it largest} solution of the equation
\beq
	\label{eq: Omegac}
 	\framebox{$\Omega_0 = \disp\frac{4}{s+2} \lf[ \gvchem - \frac{1}{2\pi} \gvm^2(0) \ri],$}
\eeq
where the r.h.s. depends on $ \Omega_0 $ through $ \gvchem $ and $ \gvm $. The existence of such a solution is proven in Proposition \ref{pro: Omegac}.
Note that thanks to the estimate $ \lf\| \gvm \ri\|^2_{\infty} \leq  \pi \gvchem $ (see \eqref{eq: gvm infty bound}), $ \Omegac > 0 $.

Before stating the main result of this paper, we have to define more precisely the region we identify with the bulk of the condensate: we set for any $ a > 0 $
\beq
	\label{eq: bulk}
	\annd : = \lf\{ \xv \in \R^2 \: | \: \gvm\lf(\tx\frac{x-1}{\eps^2}\ri) \geq |\log\eps|^{-a} \ri\},
\eeq
and observe that by the exponential decay proven in Proposition \ref{pro: gpm exp small}, $ \lf\| \gpm \ri\|_{L^2(\annd)} = 1 + o(1) $, i.e., it certainly contains the bulk of the system.


\begin{teo}[Absence of vortices in $ \annd $]
	\label{teo: no vortices}
	\mbox{}	\\
	If $ \Omega = \Omega_0/ \eps^{4} $ with $ \Omega_0 > \Omegac $ as $ \eps \to 0 $, then no GP minimizer $ \gpm $ contains vortices in $ \annd $. More precisely
	for any $ \xv \in \annd $
	\beq
		\label{eq: no vortices}
		\framebox{$\lf| \gpm(\xv) \ri|  = \disp\frac{1}{\sqrt{2\pi} \eps} \: \gvm\lf(\tx\frac{x-1}{\eps^2}\ri)  \lf(1 + \OO(\eps^{1/2} |\log\eps|^{\infty}) \ri)$.}
	\eeq
\end{teo}

\begin{rem}[Giant vortex structure]
	\mbox{}	\\
	The pointwise estimate \eqref{eq: no vortices} suggests that $ |\gpm| $ is approximately radial within $ \annd $. As already mentioned, this does not imply that the rotational symmetry is restored, since one expects that $ |\gpm| $ is far from being radial in the inner region $ x \leq \xin $, where several vortices should presumably be distributed more or less uniformly. In any case no GP minimizer is invariant under rotations if $ \Omega $ is large enough \cite[Theorem 1.6]{CPRY3}, e.g., in the giant vortex regime.
\end{rem}

\begin{rem}[Third critical velocity]
	\label{rem: sharp}
	\mbox{}	\\
	In Proposition \ref{pro: Omegac} we will prove that the equation \eqref{eq: Omegac} has a solution. Although we do not prove it, we strongly believe that such a solution is in fact unique and identifies the sharp constant in the value of the third critical speed.
	\newline
	More precisely Theorem \ref{teo: no vortices} indicates that above $ \Omegac/\eps^{4} $ the system undergoes the phase transition to the giant vortex state and the bulk of the condensate becomes vortex free. Hence 
	\beq
		\Othird \leq \frac{\Omegac}{\eps^4}.
	\eeq
	We actually expect that
	$
		\Othird =\Omegac/\eps^4
	$,
	which obviously requires to prove that the solution to \eqref{eq: Omegac} is unique. In addition one should also prove that for slower rotations vortices are still present in the bulk of the system. We plan to attack such a problem in a future work, but here we want to stress that the negativity of the cost function (see next Section \ref{sec: heuristics}) for $ \Omega_0 < \Omegac $ is a very strong indication that vortices are indeed convenient in this case and thus the sharp value of the critical speed is precisely $ \Omegac/\eps^4 $.
\end{rem}

{\begin{rem}[Comparison with \cite{CPRY3}]
	\label{rem: comparison}
	\mbox{}	\\
	We want here to discuss in more details the comparison between Theorem \ref{teo: no vortices} and the analogous result proven in \cite[Theorem 1.3]{CPRY3}: in principle, one could indeed derive an estimate of the threshold $ \bar{\Omega}_0 $ for the transition to the giant vortex state there and then it would be natural to compare it with the explicit value found here. However we provide here some heuristic arguments showing that such a comparison is actually not needed (see however next Remark \ref{rem: cost functions} for futher details).	\\
	First of all an explicit estimate of $ \bar{\Omega}_0 $ is not an easy task to achieve, due to the proof structure in \cite{CPRY3}: the result proven there is indeed obtained through an asyptotic analysis as $ \Omega_0 \to \infty $ and one should then estimate all the coefficients of the error terms appearing in the formulae. Such quantities ultimately depends on the pointwise estimate of the difference between the giant vortex profile and the ground state of the harmonic oscillator given in \cite[Proposition 3.5]{CPRY3}, which is not explicit at all. 	\\
	However, even assuming that one could obtain a sharp value $ \bar{\Omega}_0 $, there are strong reasons to believe that, unlike $ \Omegac $ (see also the previous Remark \ref{rem: sharp}), it can not be the coefficient of the critical speed. First of all the condition $ \Omega_0 > \bar{\Omega}_0 $ guarantees the positivity of the vortex energy cost in \cite{CPRY3} (Remark \ref{rem: sharp}) and therefore $ \bar{\Omega}_0 > \Omegac $. Moreover, as explained in \cite{CPRY3} (see also \cite{CPRY2}), when $ \Omega_0 \to \infty $, another transition takes place, i.e., the condensate density profile goes from a TF-like shape \eqref{eq: tfm} to a gaussian function minimizing some suitable harmonic energy. The key fact is that such a transition is expected to take place \underline{after} the giant vortex one. Indeed here we show that, for finite $ \Omega_0 $, when the profile change has not yet occurred, the condensate is already in a giant vortex state. On the opposite, a quick inspection to the proof in \cite{CPRY3} reveals that the transition to the giant vortex is proven there by imposing that the profile is already gaussian. Hence any so obtained threshold value can not be meaningful. 
\end{rem}}

\begin{rem}[Giant vortex density]
	\mbox{}	\\
	We have formulated the pointwise estimate \eqref{eq: no vortices} with $ \gvm $, but an analogous statement holds true with $ \gvm $ replaced with $ \gcrm $. The error in \eqref{eq: no vortices} is indeed so large that one can not appreciate the difference between the two reference profiles (see Proposition \ref{pro: infty est diff gvm}). Let us stress however that the use of $ \gcrm $ as a reference profile in the proof is on the opposite crucial to obtain the result (compare, e.g., the asymptotics \eqref{eq: energy asymptotics} and \eqref{eq: energy asympt gv}).
\end{rem}

	\nt
	The absence of vortices proven in Theorem \ref{teo: no vortices} and the pointwise estimate of $  \gpm$ follows from a refined result about the energy asymptotics in the same regime, that we state in the following

\begin{teo}[Energy asymptotics]
	\label{teo: energy asymptotics}
	\mbox{}	\\
	If $ \Omega = \Omega_0 \eps^{-4} $ with $ \Omega_0 > \Omegac $ as $ \eps \to  0 $, then
	\beq
		\label{eq: energy asymptotics}
		\framebox{$\gpe = \disp\frac{\gves}{\eps^4} + \OO(1).$}
	\eeq
\end{teo}

\begin{rem}[Energy expansion]	
	\mbox{}	\\
	The leading term $ \gves /\eps^{4} $ contains the main energy contribution due to the inhomogeneity of the GP profile together with the subleading kinetic energy of $ |\gpm| $. The absence of vortices in $ \annd $ can be read in the very small remainder term $ \OO(1) $. It is indeed interesting to compare \eqref{eq: energy asymptotics} with the analogous result \cite[Theorem 1.4]{CPRY3}, where the error term is much larger, i.e., $ \OO(|\log\eps|^{9/2}) $, in addition to the fact that the result proven there holds true only for $ \Omega_0 $ large enough.
	
	Notice however that the coefficient of the leading term $ \gves $ still depends on $ \eps $, through the boundaries of the integration domain as well as the optimal phase $ \betas $ and the potential $ U_{\betas} $. If one wanted to extract a proper asymptotic expansion then the natural statement would be
	\beq
		\label{eq: energy asympt gv}
		\gpe = \frac{\gve}{\eps^4} + \OO(|\log\eps|^7),
	\eeq
	with a much worse error term.
\end{rem}

	\nt
	Thanks to the pointwise statement \eqref{eq: no vortices}, one can deduce that $ \gpm $ does not vanish on $ \annd $. In particular for any $ R = 1  + \OO(\eps^2) $, $ |\gpm| > 0 $ on $ \partial \ba_R $. Hence it is possible to define the winding number of $ \gpm $ on $ \partial \ba_R $ for any such $ R $. A consequence of the energy asymptotics and the estimate \eqref{eq: no vortices} is thus the following

\begin{teo}[Winding number]
	\label{teo: winding number}
	\mbox{}	\\
	Let $ \Omega = \Omega_0 \eps^{-4} $ with $ \Omega_0 > \Omegac $ and $ R $ be any radius such that $ R = 1 + \OO(\eps^2) $ as $ \eps \to  0 $, then
	\beq
		\framebox{$ \deg \lf(\gpm, \partial \ba_R\ri) = \disp\frac{\Omega_0}{\eps^4} + \OO(1). $}
	\eeq
\end{teo}

\nt
{Note that the combination of the above result with the proof of the rotational symmetry breaking given in \cite[Theorem 1.6]{CPRY3} implies the presence of vortices in the inner hole region where $ \gpm $ is exponentially small.}


\subsection{Heuristics}
\label{sec: heuristics}

Before discussing the proofs of the main results, we briefly expose the proof strategy from a heuristic point of view, i.e., not tracking down the error terms and neglecting most technical points. As usual the main result about the behavior of the condensate wave function is deduced from the energy asymptotics \eqref{eq: energy asymptotics}. We thus focus on such a proof.

Most of the relevant features of a fast rotating Bose-Einstein condensate were already discussed in details in \cite{CPRY3} and recalled in the Introduction. Here we take as a starting point the effective functional \eqref{eq: gvbf} which is expected to provide the leading order term in the energy asymptotics in units $ \eps^{-4} $. Note that the ground state energy of $ \gvbf $ always provides an upper bound to $ \gpe $ for any {\it integer} phase, i.e., whenever $ \Omega + \beta \in \Z $. Actually the same upper bound can be proven to hold true up to some small error term even if $ \Omega + \beta $ is not an integer (see Section \ref{sec: energy ub}). Hence we can neglect the upper bound part of the proof and discuss only the lower estimate to $ \gpe $.

A preliminary step which is already described in details in \cite{CPRY3} is the restriction of the integration in $ \gpf $ to the bulk of the condensate, i.e., to an annulus centered in the origin with radius $ \simeq 1 $ and width $ \OO(\eps^2) $. This can be done by exploiting the exponential decay of $ \gpm $ outside. From now on we will then assume that the integration in $ \xv $ is restricted to the annulus $ \lf| 1 - x \ri| \leq \OO(\eps^2|\log\eps|) $.

The main steps in the energy lower bound are then the following:
\ben
	\item \underline{optimal giant vortex phase and profile}: we minimize $ \gvbe $ w.r.t. to $ \beta \in \R $ and obtain a minimizing $ \betas $ and an associated density $ \gvms $. It is crucial to observe that such a minimization yields an additional equation involving $ \gvms $, which is in fact nothing but the vanishing of the first derivative of $ \gvbe $ w.r.t. $ \beta $. Such an equation will play a crucial role at point 4 below;
	\item \underline{splitting of the energy}: using a technique introduced in \cite{LM}, which is now rather standard, we decouple $ \gpm = \frac{1}{\sqrt{2\pi} \eps}\gvms\lf(\frac{x-1}{\eps^2}\ri) u(\xv) $ and, exploiting the variational equation satisfied by $ \gvms $, we obtain
		\beq
			\gpe = \frac{\gves}{\eps^{4}} + \frac{\E[u]}{2\pi \eps^2},
		\eeq
		with $ u $ essentially minimizing the reduced energy functional
		\beq
			\label{eq: reduced energy}
			\E[u] = \int \mathrm{d}\xv \: \gvms^2 \left\{
			\tx\frac{1}{2} {|\nabla u|} ^2 + {\bf a}(\xv) \cdot {\bf j}_u(\xv)
			+\frac{1}{2\pi\varepsilon ^4} \gvms ^2{(1-|u| ^2)} ^2
		\right\},
		\eeq
		where the ``magnetic potential'' $ \av $ depends on $ \Omega $ and $ \betas $ and $ {\bf j}_u $ is the {\it superconducting current}
		\beq
			\label{eq: current}
			{\bf j}_u(\xv) = \tx\frac{i}{2} \lf( u \nabla u^* - u^* \nabla u \ri).
		\eeq
		Completing the lower bound means to show that $ \E[u] $ is positive;
	
	\item \underline{hydrodynamic estimate}: we note that the ``magnetic potential'' is divergence free and therefore it exists a {\it potential function} $ F(x) $ such that $ 2\gvms^2(x) \av(\xv) = - \nabla^{\perp} F(x) $. This trick was first used in \cite{CRY} in the context of the GP theory for rotating condensates. For later applications to the GL function see also \cite{CR2,CR3}. We can thus integrate by parts the second term in \eqref{eq: reduced energy} obtaining 
	\beq
		\label{eq: angular momentum term}
		\int \diff \xv \: F(x) \: \curl \lf( {\bf j}_u \ri).
	\eeq
	At this stage we observe that since $ \betas = \OO(1) $ and it appears in \eqref{eq: gvbf} always multiplied by $ \eps^2 $, a good approximation of the functional $ \gvf_{\betas} $ can be obtained by taking the limit $ \eps \to 0 $, which yields the functional \eqref{eq: gvf}, with ground state energy $ \gve $ and minimizer $ \gvm $. We can also replace $ F(x) $ with its limiting counterpart $ \fgv(x) $, which is in fact a negative function. The last step to estimate \eqref{eq: angular momentum term} is to use the trivial inequality $ |\curl \lf( {\bf j}_u \ri)| \leq \lf| \nabla u \ri|^2 $ and the negativity of $ \fgv $ to get the lower bound
	\beq
		\E[u] \geq \int \mathrm{d}\xv \: 
			\lf(\tx\frac{1}{2} \gvm^2 + \fgv \ri) {|\nabla u|} ^2,
	\eeq
	where we have also dropped the last positive term in \eqref{eq: reduced energy};

	\item \underline{positivity of the cost function}: the above lower bound suggests that any topological defect of $ u $ should carry an energy cost given by the {\it cost function}
		\beq
			\label{eq: kgv}
			\kgv = \half \gvm^2 + \fgv.
		\eeq
		Positivity of such a function in the bulk would then imply that vortices are not energetically favorable anywhere in the condensate. This is turn can be proven by direct inspection of the function itself. First we observe that both $ \gvm $ and $ \fgv $ are radial functions and we therefore change coordinates $ x = 1 + \eps^2 y $, so that in the new variable $ y $ the bulk of the condensate is basically the whole real line. In the new variable the explicit expression of $ \fgv $ (that we still denote by $ \fgv $) is
		\beq
			\label{eq: fgv}
			\fgv(y) = -2\Omega_0 \int_y^{\infty} \diff t \: t \: \gvm^2(t).
		\eeq
		Notice that by symmetry\footnote{Unlike $ \gvm $, the profile $ \gvms $ is not exactly symmetric, but $ F $ satisfies analogous properties thanks to the optimality condition of $ \betas $, i.e., the additional equation involving $ \gvms $ and $ \betas $ which was mentioned at point 1.}  of $ \gvm $, $ \fgv(-\infty) = 0 $ and $ \fgv(y) \leq 0 $ for any $ y \in \R $. The cost function $ \kgv $ is therefore smooth and $ \kgv(\pm \infty) = 0 $, so that, if it becomes negative, it must have a minimum. The derivative of $ \kgv $ can be easily computed
		\beq
			{\kgv}^{\prime}(y) = \gvm(y) \gvm^{\prime}(y) + 2\Omega_0 y \gvm^2(y),
		\eeq
		so that, by strict positivity of $ \gvm $, at any critical point $ y_0 $ for $ \kgv $, one has 
		\beq
			\label{eq: critical point}
			\gvm^{\prime}(y_0) = - 2\Omega_0 y_0 \gvm(y_0).
		\eeq		
		Now using the variational equation for $ \gvm $ and manipulating the expression \eqref{eq: fgv} of the potential function, it is possible to show that the cost function can be equivalently rewritten as
		\beq
			\kgv = \lf[\frac{1}{2}  + \Om _0 y ^2 + \frac{\Om _0}{\pi \alpha ^2} \gvm^2 - \frac{2 \Om _0 \gvchem}{\alpha ^2} \ri] \gvm ^2 - \frac{\Om _0}{\alpha ^2}{\gvm'}^2 
		\eeq
		and, inserting the condition \eqref{eq: critical point} satisfied at any minimum point $y_0 $ of $ \kgv $, we get
		\beq
			\kgv(y_0) = \lf[\frac{1}{2}  + \frac{\Om _0(s+1)}{s+2} y_0 ^2 + \frac{\Om _0}{\pi \alpha ^2} \gvm^2(y_0) - \frac{2 \Om _0 \gvchem}{\alpha ^2} \ri] \gvm ^2(y_0).
		\eeq
		Using the parity of $ \gvm $ as well as the variational equation, one can prove that the quantity between brackets on the r.h.s. of the expression above is positive if and only if it is positive at the origin (see Proposition \ref{pro: kgv positive})
		\beq
			\frac{1}{2}  +  \frac{\Om _0}{\pi \alpha ^2} \gvm^2(0) - \frac{2 \Om _0 \gvchem}{\alpha ^2} = \frac{1}{2}  +  \frac{2}{\Omega_0 (s+2)} \lf[ \frac{1}{2\pi} \lf\| \gvm \ri\|_{\infty}^2 -  \gvchem \ri] \geq 0	\quad	\Longleftarrow \quad \Omega_0 \geq \Omegac.
		\eeq
\een

Once the energy asymptotics is proven, the pointwise estimate of $ |\gpm|$, which allows to exclude the presence of vortices in the bulk for $ \Omega_0 > \Omegac $, is a simple consequence: putting back the positive term we have dropped in the lower bound, one first obtains an estimate of the region where $ |u| $ can differ from $ 1 $. Then combining this with an $ L^{\infty} $ estimate of the gradient of $ u $, one gets the result.

	It is worth mentioning at this stage a technical difference with previous approaches. Indeed in \cite{CPRY3} two potential functions were actually used instead of one, in order to get rid of boundary terms coming from the integration by parts described at  step 3 (see the discussion in \cite[Sect. C]{CPRY3}). Here on the opposite we are able to use only one potential function by estimating in a more refined way the boundary terms (compare, e.g., with next \eqref{eq: int by parts}). As in \cite{CPRY3} we also exploit the symmetry properties of the profile $ \gvms $, which is to a very good approximation invariant under reflections w.r.t. the origin.

\section{Preliminary Estimates}
\label{sec: preliminary}

Here we collect some useful technical results as well as the main properties of the effective functionals involved in the analysis. An important piece of information is contained in Section \ref{sec: positivity} where we prove the positivity of the cost function.

\subsection{Giant Vortex Functionals}
\label{sec: gv profiles}

We start by describing the derivation of the functional \eqref{eq: gvbf} from the GP energy. As anticipated in Section \ref{sec: main results} the idea is to evaluate the energy of a trial state of the form $ f(x) e^{i n \vartheta} $ in polar coordinates $ \xv = (x, \vartheta) $ and with $ n = \Omega + \beta $. In addition we assume that $ f $ is real as it will be for any giant vortex profile. The result of a rather simple computation is 
\beq
	\label{eq: proto gvf}
	\gpf[f e^{i n \vartheta}]
	= 2\pi \int_0^{\infty} \diff x \: x  \lf\{
		\half \lf|
			\nabla f
		\ri| ^2
		+ \Omega^2 \lf[
			\tx\frac{1}{2}\lf(
				x ^2
				-\frac{\Omega + \beta}{\Om}
			\ri) ^2
			\frac{1}{x ^2}
			+W(x)
		\ri] f ^2
		+ \frac{1}{\eps^2} f ^4
	\ri\}.
\eeq
Exploiting the exponential smallness of $ \gpm $ outside of the bulk of the condensate proven in \cite{CPRY3} and recalled in next Proposition \ref{pro: gpm exp small}, we can restrict the integration domain to the annulus
\beq
	\label{eq: ring}
	\ann:
	=\lf\{
		\xv \in \RR :\:
		\lf|
			1
			-\xv
		\ri|
		\le \eps ^2 \eta
	\ri\},\qquad
	\eta:
	= \tx\frac{\eta _0}{2\sqrt{\Om _0}}|\log\eps|,
\eeq
where $ \eta_0 >  0 $ is an arbitrary finite constant and the prefactor in the definition of $ \eta $ has been chosen of that form for further convenience. Thanks to \eqref{eq: gpm exp small}
\beq
	\label{eq: change of coordinates}
	\gpm(\xv) = \OO(\eps^{\infty}),	\qquad \mbox{for any } \xv \notin \ann,
\eeq
and the restriction is thus well motivated. In addition we will also see that a similar estimate holds true for any giant vortex profile. In terms of the one-dimensional functional \eqref{eq: proto gvf} we are then integrating in the interval $ [1-\eps^2 \eta,1+ \eps^2 \eta] $ and a change of variable is called for: setting 
\beq
	x = 1 + \eps ^2 y,	\qquad		g(y) = \sqrt{2\pi}\eps \: f(1 + \eps^2y)
\eeq
so that $ g $ is normalized in\footnote{We set in fact $ L^p_{\eta} : = L^p([-\eta,\eta],(1 + \eps^2 y) \diff y) $ for any $ 1 \leq p \leq \infty $.} $ L^2_{\eta} : = L^2([-\eta,\eta], (1 + \eps^2y) \diff y) $, we obtain the energy
\bml{
	\label{eq: quasi gvf}
	\tgpf[g]
	= \frac{1}{\varepsilon ^4}
	\int _{-\eta}^{\eta}
	\diff y \: (1 + \eps^2 y)
	\lf\{
		\tx\frac{1}{2}
		\lf(
			g^{\prime} 
		\ri)^2 +
	\ri.
	\\
	\lf.
		+\eps ^4 \Om ^2
		\lf[
			\tx\frac{1}{2}\lf(
				1 + \eps^2 y
				-\frac{\Om + \beta}{\Om\xinyp}
			\ri) ^2
			+W(\xiny)
		\ri]g ^2
		+ \tx\frac{1}{2\pi}
		g ^4
	\ri\}.
}
Now we expand $ W(1 + \eps^2 y) $ in Taylor series around $ y = 0 $ to get
\beq
	\label{eq: taylor}
	W(1 + \eps^2 y)
	= \tx\frac{s-2}{2} \eps^4 y^2	+ \tx\frac{(s-2)(s-1)}{6} \eps^6 y^3 + \eps^8 \varphi(y)
\eeq
where $ \varphi (y) = \OO(y^4) $.
Using this fact we can rewrite the potential in \eqref{eq: quasi gvf} as (recall that $ \alpha^2 = \Omega_0^2 (s + 2) $)
\beq
	\eps ^4 \Om ^2
		\lf[
			\tx\frac{\lf( 2 \eps^2 y - \eps^4 \beta/\Omega_0 + \eps^4 y^2\ri)^2}{2(1 + \eps^2 y)^{2}}
			+W(\xiny)
		\ri] = U_{\beta}(y) + \eps^2 y^3 v(y),
\eeq
with $ v $ independent of $ \beta $ and of lower order w.r.t. to $ U_{\beta} $. Explicitly
\beq
	\label{eq: pot U}
	U _\beta(y):
	= \frac{1}{\lf(\xiny\ri) ^2}
	\lf(
		\frac{\alpha ^2}{2} y ^2
		-2 \Om _0 \eps ^2 \beta y
		-\Om _0 \eps ^4 \beta y ^2
		+\frac{1}{2}\eps ^4 \beta ^2
	\ri),
\eeq
\beq
	\label{eq: pot v}
	v(y):
	= \frac{\Om _0 ^2(s + (s-1)\eps ^2 y)}{\lf(\xiny\ri) ^2}
	+\frac{(s-1)(s-2) \Omega_0^2}{6}
	+\frac{\eps^2 \Om _0 ^2}{y ^3} \varphi\lf(y\ri).
\eeq
Some trivial estimate using the Taylor expansion \eqref{eq: taylor} implies that for $ y \in [-\eta,\eta] $
\beq
	\label{eq: taylor U}
	U _\beta(y) = \tx\frac{1}{2} \alpha^2 y^2 + \OO\lf(\eps^2 (1 + |\beta|) \eta + \eps^4 \beta^2\ri),
\eeq
which shows that, if, e.g., $ \beta $ is uniformly bounded in $ \eps $, the potential $ U_{\beta}(y) $ is harmonic up to corrections of higher order in $ \eps $. Alternatively one can think of $ U_{\beta} $ as a shifted harmonic oscillator by writing
\beq
	U_{\beta}(y) = \tx\frac{1}{2} \alpha^2 \lf( y - \tx\frac{2 \Omega_0 \eps^2 \beta}{\alpha^2} \ri)^2 + \OO(\eps^4 |\beta| \eta^2 + \eps^4 \beta^2).
\eeq
In fact, since the optimal value of $ \beta $ we are going to choose is $ \OO(1) $, both representations are equivalent since the shift will be $ \OO(\eps^2) $.
Concerning the rest $ v(y) $ one trivially has the upper bound
\beq
	\label{eq: taylor v}
	\lf| v(y) \ri| \leq C_{\Omega_0} + \OO(\eps^2 \eta),
\eeq
for $ y \in [-\eta,\eta] $ and with a finite constant $ C_{\Omega_0} $. The rest in the above expression is a consequence of the bound $ |\varphi(y)| \leq C |y|^4 $, which follows from the Taylor expansion \eqref{eq: taylor}.

In conclusion we have recovered the expression \eqref{eq: gvbf}, i.e.,
\bdm
	\gvbf[g] = \int _{-\eta}^{\eta} \diff y \: (1 + \eps^2 y) \lf\{ \half
		\lf(
			g^{\prime}
		\ri)^2
		+U _\beta(y) g ^2 
		+\varepsilon ^2 y ^3 v(y) g ^2
		+\frac{1}{2\pi}
		g ^4
	\ri\}.
\edm
We now discuss the minimization of such a function w.r.t. $ g $ and for that purpose we have to identify the proper minimization domain, i.e.,
\beq
	\label{eq: gvbeta domain}
	\gvbdom:
	=\lf\{
		g \in H ^1 (-\eta,\eta)\: \big|\: 
		g = g ^*,\:
		\lf\|
			g
		\ri\| _{L ^2_{\eta}}
		= 1
	\ri\}.
\eeq
The ground state energy of $ \gvbf $ is defined as
\beq
	\label{eq: gvbeta energy}
	\gvbe =\inf _{g \in \gvbdom} \gvbf[g].
\eeq
Notice that the assumption $ g = g^* $, i.e., reality of the argument, does not imply any loss of generality because the ground state can always be chosen real (see next Proposition).

\begin{pro}[Minimization of $ \gvbf $]
	\label{pro: min gvbf}
	\mbox{}	\\
	There exists a minimizer $ \gvbm \in \gvbdom $ of \eqref{eq: gvbf} that is unique up to a sign, radial
	and can be chosen strictly positive. 	In addition $ \gvbm \in C^{\infty}(-\eta,\eta) $ and it solves the variational equation
	\beq
		\label{eq: gvbm variational}
		-\tx\frac{1}{2}\gvbm''
		-\frac{\eps ^2}{2\xinyp}\gvbm'
		+U _\beta (y) \gvbm
		+\varepsilon ^2 y ^3 v(y)\gvbm
		+\frac{1}{\pi} \gvbm ^3
		=\gvbchem \gvbm
	\eeq
	with Neumann boundary conditions $ \gvbm^{\prime}(\pm \eta) = 0 $ and $ \gvbchem = \gvbe + \frac{1}{2\pi} \|\gvbm\| _{L ^4_{\eta}} ^4 $. 	\\
	Finally $ \gvbm $ has a unique maximum point at $ \yb  $ and it decreases monotonically anywhere else.

\end{pro}

\begin{proof}
	Existence and uniqueness of the minimizer follow from strict convexity of the functional $ \gvbf [\sqrt{\rho}] $ with respect to the density $ \rho = g ^2 $. The variational equation \eqref{eq: gvbm variational} is satisfied at least in weak sense. Then one deduces the strict positivity of $ \gvbm $ noticing that it is actually a ground state of a suitable one-dimensional Schr\"odinger operator. The equality for $ \gvbchem $ follows integrating the \eqref{eq: gvbm variational} and recalling the fact that $ \gvbm $ has $ L ^2 $-norm equal to one. Finally a trivial bootstrap argument allows to deduce smoothness of $ \gvbm $ and therefore that \eqref{eq: gvbm variational} is solved in a classical sense.
	
	The only non-trivial result is the one about the existence of the a single maximum point for $ \gvbm $. However it follows from the property of the potential $ U_{\beta}(y) + \eps^2 y^3 v(y) $: going back to the expression of the potential in \eqref{eq: proto gvf}, one can easily compute, with $ x = 1  + \eps^2 y $,
	\bdm
		\frac{\partial \lf[ U_{\beta}(y) + \eps^2 y^3 v(y) \ri]}{\partial x} = \frac{1}{x^3} \lf[ x^{s+2} - \lf(1 + \tx\frac{\eps^4 \beta}{\Omega_0} \ri)^2 \ri],
	\edm
	which vanishes at a single point $ \ypot $, i.e., where 
	\bdm
		1 + \eps^2 \ypot = \lf(1 + \tx\frac{\eps^4 \beta}{\Omega_0} \ri)^{\frac{2}{s+2}} = 1 + \tx\frac{2 \beta \eps^4}{(s+2) \Omega_0} + \OO(\eps^8 \beta^2).
	\edm 
	The Taylor expansion also shows that
	\bdm
		\ypot = \tx\frac{2 \beta \eps^2}{(s+2) \Omega_0} (1 + \OO(\eps^4 \beta)).
	\edm
	The monotonicity property of $ \gvbm $ can then be obtained by a simple rearrangement argument (see, e.g., \cite[Proposition 2.2]{CPRY1}): since the potential has a single maximum point, if $ \gvbm $ had more than one maximum besides $ \yb $, one could move mass from the further maximum to the minimum in between and lower the energy. Since $ \gvbm $ is a minimizer one gets a contradiction.	
\end{proof}





\nt
Another effective one-dimensional functional which is going to play an important role in the analysis is \eqref{eq: gvf}, i.e., the formal limit $ \eps \to 0 $ of $ \gvbf $, assuming that $ \beta = o(\eps^{-2}) $:
\bdm
	\gvf[g]
	=\int _{\R}
	\diff y
	\lf\{
		\tx\frac{1}{2}\lf(
			g'
		\ri)^2
		+\frac{\alpha ^2}{2} y ^2 g ^2
		+\frac{1}{2\pi} g ^4
	\ri\}.
\edm
The minimization domain is in this case given by
\beq
	\gvdom:
	=\lf\{
		g \in H ^1(\mathbb{R})\:\big|\:
			g
			= g ^*,
			\lf\| g \ri\| _{L ^2(\R)}
			= 1
		\ri\},
\eeq
and the ground state energy will be denoted by $ \gve = \inf _{g \in \gvdom} \gvf[g] $.

\begin{pro}[Minimization of $ \gvf $]
	\label{pro: min gvf}
	\mbox{}	\\
	There exists a minimizer $ \gvm \in \gvdom $ of \eqref{eq: gvf} that is unique up to a sign, radial
	and can be choose strictly positive. In addition $ \gvm \in C^{\infty}(\R) $ and it solves the variational equation
	\beq
		\label{eq: gvm variational}
		-\tx\frac{1}{2}\gvm'' 
		+\frac{\alpha ^2}{2} y^2 \gvm 
		+\frac{1}{\pi} \gvm ^3 
		=\gvchem \gvm 
	\eeq
	with $ \gvchem = \gve + \frac{1}{2\pi} \|\gvm\| _4 ^4 $.	\\	
	Finally $ \gvm $ is even w.r.t. the origin and has only one maximum at $ y = 0 $, which fulfills the inequality
	\beq
		\label{eq: gvm infty bound}
		\gvm^2(0) = \lf\|\gvm\ri\| _\infty ^2 \leq \pi \gvchem.
	\eeq	
\end{pro}

\begin{proof}
	See the proof of Proposition \ref{pro: min gvbf}. Parity of $ \gvm$  is a trivial consequence of the parity of the potential. The inequality \eqref{eq: gvm infty bound} follows from direct inspection of the variational equation \eqref{eq: gvm variational}: at any maximum point $ \gvm'' \leq 0 $, which immediately implies the result.
\end{proof}

\subsection{Estimates of the Gross-Pitaevskii and Giant Vortex Profiles}

In this Section we collect several technical estimates of the profiles involved in the discussion. Such estimates will play a key role in the proofs but can be typically obtained by standard techniques in functional analysis.

We start by recalling a result which was in fact proven in \cite[Propositions 3.1 and 3.2]{CPRY3}: {let $ \eta_0 $ be the parameter appearing in the definition \eqref{eq: ring} of $ \ann $, then}

	\begin{pro}[Exponential decay of $ \gpm $]
		\label{pro: gpm exp small}
		\mbox{}	\\
		If $ \Omega = \Omega_0 / \eps^4 $, there exists two finite constants $ c, C > 0 $ (independent of $ \eta_0 $) such that, for any $ \xv \notin \ann $, 
		\beq
			\label{eq: gpm exp small}
			\lf| \gpm(\xv) \ri|^2 \leq \frac{C}{\eps^2} \max\lf[ \eps^{\frac{c \eta_0^2}{4}}, \exp\lf\{ - \tx\frac{\sqrt{\Omega_0}}{\eps^2} |1 - x|\ri\} \ri].
		\eeq
	\end{pro}
	
	\nt
	In particular the above result implies that by taking $ \eta_0 $ large enough we can make $ \gpm $ arbitrarily small outside $ \ann $. This fact will be crucial in restricting the computation of the GP energy within $ \ann $. Notice also that as soon as $ |1 - x| \gg \eps^2|\log\eps| $, $ \gpm = \OO(\eps^{\infty}) $.
	
	Let us now focus on the giant vortex profiles. Before stating the main technical estimates we first formulate a simple preliminary bound on the giant vortex energy $ \gvbe $:



	\begin{pro}[Preliminary bound on $ \gvbe $]
		\label{pro: estimates on energy}
		\mbox{}	\\
		If $ \beta = \bigO{\eps ^{-2}} $ as $ \eps \to 0$, then 
		\beq
			\label{eq: en preliminary est}
			\gvbe = \OO(1),	\qquad
			\gvbchem
			= \OO(1).
		\eeq
	\end{pro}

\begin{proof}
	Since $ \gvbe $ is positive (compare with \eqref{eq: quasi gvf}), it suffices to prove a suitable upper bound: to that purpose one can simply evaluate the functional $ \gvbf $ on the ground state of the one-dimensional harmonic oscillator with frequency $ \alpha $. The result on $ \gvbchem $ follows from the trivial estimates $ \gvbe \leq \gvbchem \leq 2 \gvbe $.
\end{proof}


\nt
The giant vortex profile $ \gvbm $ decays exponentially for large $ |y|$ and one can actually show that this decay captures the correct asymptotics of $ \gvbm $:

	\begin{pro}[Pointwise estimates of $ \gvbm $]
		\label{pro: gvbm point est}
		\mbox{}	\\
		If $ \beta = \bigO{\eps ^{-2}} $ as $ \eps \to 0 $, then there exists a finite constant $ C $ such that
		\beq
			\label{eq: estimate on gvbm}
			\gvbm(y)
			\le C e ^{-2\sqrt{\Om _0}|y|},	\qquad		\mbox{so that } \gvbm(\pm \eta) =\bigO{\eps ^{\eta _0}}.
		\eeq
		If  $ \beta = \bigO{1} $ then there exist two finite constants $ C _1, C _2 > 0 $ such that the following inequalities hold true
		\beq
			\label{eq: better estimate on gvbm}
			C_1 \| \gvbm \| _{L _\eta ^4} ^2
			\exp\lf\{
				- \tx\frac{\alpha}{2} y ^2
			\ri\}
			\le \gvbm(y)
			\le C_2 
			\exp\lf\{
				- \tx\frac{\alpha}{4} y ^2
			\ri\}.
		\eeq
	\end{pro}

	\begin{proof}
		The results are proven by means of standard super- and sub-solution techniques. We spell however the proofs in full details for the sake of clarity.
	
		To prove \eqref{eq: estimate on gvbm} it is somehow more convenient to go back to the variational equation satisfied by $ f_{\beta}(x) = (\sqrt{2\pi}\eps)^{-1} \: \gvbm((x-1)/\eps^2) $, i.e.,
		\beq
			\label{eq: vareq for f}
			-\tx\frac{1}{2} f^{\prime\prime}_\beta - \frac{1}{2r} f^{\prime}_\beta + \tx\frac{\Om ^2}{2x ^2} \lf (
			x ^2
			- \frac{\Omega + \beta}{\Om}
			\ri) ^2
			f _\beta
			+\Om ^2 W (x) f _\beta
			+\frac{1}{\eps ^2} f _\beta ^3
			=\frac{1}{\eps ^4} \gvbchem f _\beta.
		\eeq
		The first simple observation is that by positivity of $ f_{\beta} $ and $ W(x) $, we get
		\[
			-\tx\frac{1}{2} f^{\prime\prime}_\beta - \frac{1}{2r} f^{\prime}_\beta
			\le \frac{1}{\eps ^4}
			\lf(
			\mu _\beta
			-\eps ^2 f _\beta ^2
			\ri) f _\beta,
		\]
		which, by negativity of the second derivative of $ f_\beta $ at any maximum point of $ f_{\beta} $, immediately implies the upper bound
		\bdm
			\lf\| f_\beta \ri\|_{L^\infty(\ann)}^2 \leq \tx\frac{1}{\eps^2} \gvbchem,
		\edm
		which in terms of $ \gvbm $ becomes, via \eqref{eq: en preliminary est} (here we are assuming that $ \beta = \OO(\eps^{-2}) $),
		\beq
			\label{eq: gvbm infty est}
			\lf\|\gvbm \ri\|_{L^\infty_{\eta}} = \OO(1).
		\eeq
	
		In order to prove \eqref{eq: estimate on gvbm} we will provide an explicit supersolution to the equation \eqref{eq: vareq for f}. Notice that the first two terms of the equation form the two-dimensional Laplacian, i.e., for any radial function $ f $, $ - \Delta f = - \frac{1}{r} \partial_r ( r f^{\prime}) $. We will use this fact to construct a supersolution in dimension two. Let then $ a > 0 $ be a parameter independent of $ \eps $ that is going to be chosen later and consider the two-dimensional region
		\beq
			\mathcal{A}:
			= \ba _{1 - a \eps ^2}(0)
			\cap \ann
			=\lf\{
			\xv \in \RR
			\;\big|\;
			1 - \eta \eps ^2
			\leq x
			\leq 1 - a \eps ^2
			\ri\}.
		\eeq
		Inside $ \A $ one has the lower bound
		\[
			W (x)
			\ge \tx\frac{s-2}{2}\lf(
			x-1
			\ri) ^2
			+\bigO{\lf|x-1\ri| ^3}
			\ge \frac{\lf(s-2\ri)}{2} a ^2 \eps ^4
			+\bigO{\eta ^3 \eps ^6}
			\ge C _0 a ^2 \eps ^4
		\]
		with $ C _0 > 0 $, so that \eqref{eq: vareq for f} and \eqref{eq: en preliminary est} yield
		\[
			- \tx\frac{1}{2} \Delta f _\beta
			\le \frac{1}{\eps ^4} \gvbchem f _\beta
			-\Om ^2 W(x) f _\beta
			\le \frac{1}{\eps ^4}
			\lf(
			C _1
			-\Om _0 ^2 C _0 a ^2
			\ri)f _\beta
		\]
		where also $ C_1 > 0 $. If now we pick $ a ^2 \ge \frac{2 C _1}{C _0 \Om _0 ^2} $, we get that $ f _\beta $ is a subsolution of the following differential problem
		\beq
			\label{eq: f subsol}
			-\tx\frac{1}{2} \Delta f
			+\frac{1}{2} C _0 a ^2 \Om ^2 \eps ^4 f
			=0.
		\eeq
		To get rid of the inner boundary we now extend $ f_\beta $ to the whole ball $ \ba_{1 - a\eps^2}(0) $ in a smooth (in fact at least $ C^2 $) way. We denote by $ \tilde{f} $ such a new function and we require that $ \tilde{f}(x) = 0 $ for $ x \leq 1 - 2 \eps^2 \eta $ and
		\beq
			 -\tx\frac{1}{2} \Delta f
			+\frac{1}{2} C _0 a ^2 \Om ^2 \eps ^4 f
			\leq 0,
		\eeq
		for any $ \xv \in \ba_{1 - a\eps^2}(0) $. We omit the explicit details of such a construction for the sake of brevity.
		A supersolution to the same problem can be constructed by taking $$ f _\mathrm{sup}(x):=C _a \| f _\beta \| _\infty e ^{-\sqrt{\Om}\lf(1-x ^2\ri)}$$ with $ C _a $ a constant to be suitably chosen:
		\[
			-\tx\frac{1}{2} \Delta f _\mathrm{sup}
			+\frac{1}{2} C _0 a ^2 \Om ^2 \eps ^4 f _\mathrm{sup}
			=\frac{1}{2}\lf(
			-2\sqrt{\Om}
			-2 \Om x ^2
			+ C _0 \Om _0 a ^2 \Om
			\ri)  f _\mathrm{sup} >0,
		\]
		if we choose $ a ^2 > \frac{4}{C _0 \Om _0} $. The constant $ C _a $ is then used to guarantees that $ f_{\mathrm{sup}} $ satisfies the proper boundary conditions. In order to apply the maximum principle (see, e.g., \cite[\S~6.4.1, Theorem 2]{E}), we need that $ \tilde{f}(x) \le f _\mathrm{sup} (x) $ on $ \partial \ba_{1 - a\eps^2}(0) $, which holds true if  $ C _a \geq e ^{2\sqrt{\Om _0}a} $:
		\[
			\lf.
				f _\mathrm{sup}
			\ri| _{\partial \ba_{1 - a\eps^2}(0)}
			= C _a \| f _\beta \| _\infty e^{-\sqrt{\Om _0}a \lf(2-a\eps ^2\ri)}
			\ge \lf.
				f_{\beta}
			\ri| _{\partial \ba_{1 - a\eps^2}(0)}.
		\]
		Hence we conclude that $ \tilde{f} \le f _\mathrm{sup} $ in the whole $ \ba_{1 - a\eps^2}(0) $, and therefore, using the monotonicity of $ f _\mathrm{sup} $, $ f_\beta \leq  f _\mathrm{sup} $ in the whole region $ \ba _1 (0) \cap \ann $. Going back to $ \gvbm $ and using \eqref{eq: gvbm infty est}, we obtain \eqref{eq: estimate on gvbm} in $ \ba _1 (0) \cap \ann $. To extend the result to the complementary region, one can use a very similar argument with the trivial change $ x^2 - 1 \to 1 - x^2 $ in the supersolution.
	

		For the refined estimates \eqref{eq: better estimate on gvbm}, we consider the variational equation \eqref{eq: gvbm variational} for $ a \leq |y| \leq \eta $, with $ a > 0 $ such that $ a ^2 > \frac{8\gvbchem}{3\alpha ^2} $ and $ \eps $ small enough, which imply
	%
		\[
			U _\beta (y)
			+\eps ^2 y ^3 v(y)
			-\gvbchem
			= \tx\frac{\alpha ^2}{2} y ^2
			-\gvbchem
			+\bigO{\eps ^2 \eta ^3}
			\ge \frac{\alpha ^2}{8} y ^2
		\]
		and therefore in that region $ \gvbm $ is a subsolution of the equation
		\beq
			\label{eq: subequation}
			-\tx\frac{1}{2} g''
			-\frac{\eps ^2}{2\xinyp} g'
			+\frac{\alpha ^2}{8} y ^2 g = 0.
		\eeq
		As before we extend $ \gvbm $ to the whole region $ |y| \geq a $ in a $C^2$ way and preserving the differential inequality satisfied in $ a \leq |y| \leq \eta $, i.e.,
		 \bdm
			-\tx\frac{1}{2} g''
			-\frac{\eps ^2}{2\xinyp} g'
			+\frac{\alpha ^2}{8} y ^2 g \leq 0.
		\edm
		Again we skip the details for brevity. 
	
		Now for some $ C > 0 $ to be fixed later the following function
		\[
			g _\mathrm{sup} (y):
			=C e ^{-\frac{\alpha}{4}y ^2}
		\]
		is a supersolution to \eqref{eq: subequation}: for $ \eps $ small enough
		\[
			-\tx\frac{1}{2} g _\mathrm{sup}''
			-\frac{\eps}{2\xinyp} g _\mathrm{sup}'
			+\frac{\alpha ^2}{8} y ^2g _\mathrm{sup}
			=\lf(
			\frac{\alpha}{4}
			+\frac{\alpha y \eps ^2}{4\xinyp}
			\ri)g _\mathrm{sup}
			\geq 0.
		\]
		Choosing the $ C \geq \|\gvbm\| _\infty e ^{\frac{\alpha a ^2}{4}} $ to ensure that $ \gvbm(\pm a) \leq g_{\mathrm{sup}}(a) $, we get the upper estimate.
	%
	%
	%
	
		Analogously we can choose $ C > 0 $ in such a way that 
		\[
			g_{\mathrm{sub}}:
			=C e ^{-\frac{\alpha}{2}y ^2}
		\]
		is a subsolution to \eqref{eq: gvbm variational}: first one notes that
		\[
			\tx\frac{\alpha}{2}
			\le \gvbchem
			-\frac{1}{\pi} \| \gvbm \| _{L _\eta ^4} ^4
			+\bigO{\eps ^{2} \eta},
		\]
		which follows from the fact that the harmonic oscillator on the real line is bounded from below by $ \alpha/2 $; then using this inequality in \eqref{eq: gvbm variational}, we obtain
		\bmln{
			 -\tx\frac{1}{2} g_{\mathrm{sub}}''
	 		-\frac{\eps ^2}{2\xinyp} g_{\mathrm{sub}}'
			 +U _\beta (y) g_{\mathrm{sub}}
			 +\eps ^2 y ^3 v(y) g_{\mathrm{sub}}
			 +\frac{1}{\pi} g_{\mathrm{sub}}^3 
			 - \gvbchem g_{\mathrm{sub}}	\\
			=\lf[
			\tx\frac{\alpha}{2}
			-\gvbchem
			+\frac{1}{\pi} g_{\mathrm{sub}}^2
			+U _\beta (y)
			-\frac{\alpha ^2}{2}y ^2
			+\frac{\alpha y \eps ^2}{2\xinyp}
			+\eps ^2 y ^3 v(y)
			\ri] g_{\mathrm{sub}}	\\	
			\le \lf[
			\tx\frac{1}{\pi}\lf(
				 g _\mathrm{sub} ^2 (y)
				 -\|\gvbm\| _{L _\eta ^4} ^4
			\ri)
			+\bigO{\eps ^2 \eta ^3}
			\ri] g _\mathrm{sub}
			< 0,
		}
		if we pick $ C < \|\gvbm\| _{L _\eta ^4} ^2 $. To conclude we use the fact that $  g_{\mathrm{sub}} $ goes to 0 as $ |y| $ goes to infinity: indeed it is sufficient to observe that there certainly exists a point $ \bar{y} >0  $ such that $  g_{\mathrm{sub}} (\pm\bar{y}) = \min \lf\{\gvbm (\eta), \gvbm (-\eta) \ri\} $ and
		\[
			\widetilde{g}(y):
			=\lf\{
				\begin{array}{ll}
					\gvbm(y) & |y| \le \eta, \\
					\gvbm(\eta) & \eta \le y \le \bar{y}, \\
					\gvbm(-\eta) & -\bar{y} \le y \le -\eta,
				\end{array}
			\ri.
		\]
		 is a supersolution to \eqref{eq: gvbm variational}, satisfying $ \tilde{g}(\pm \bar{y}) \geq g_{\mathrm{sub}}(\pm \bar{y}) $. Hence $ g_{\mathrm{sub}} \leq \tilde{g} $ for any $ |y| \leq \bar{y} $, which implies the lower estimate \eqref{eq: better estimate on gvbm} for $ |y| \leq \eta $.
	\end{proof}

\nt
We conclude this Section by stating analogous pointwise estimate for the limiting profile $ \gvm $:



	\begin{pro}[Pointwise estimates of $ \gvm $]
		\label{pro: point est gvm}
		\mbox{}	\\
		There exists a finite constant $ C >0 $ such that
		\beq
			\label{eq: estimate on gvm}
			\|
				\gvm
			\| _4 ^2 \exp\lf\{-\tx\frac{\alpha}{2}y ^2\ri\}
			\le \gvm(y)
			\le C \exp\lf\{-\tx\frac{\alpha}{4}y ^2\ri\}.
		\eeq
	\end{pro}

	\begin{proof}
		The estimate can be proven exactly as \eqref{eq: better estimate on gvbm} in Proposition \ref{pro: gvbm point est} and we skip the details.
	\end{proof}

\subsection{Optimal Giant Vortex Phase and Profile}
\label{sec: betas}

In this Section we investigate the minimization of $ \gvbe $ w.r.t. $ \beta \in \R $. The main result is the following

	\begin{pro}[Optimal phase]
		\label{pro: betas}
		\mbox{}	\\
		For $ \eps $ small enough there exists a unique minimizer $ \betas \in \R $ such that
		\beq
			\label{eq: minimal energy}
			\gcre
			: =\inf _{\beta \in \mathbb{R}} \gvbe = \gve_{\betas}.
		\eeq
	 	Such an optimal phase is explicitly given by
		\beq
			\label{eq: betas}
			\betas
			= - \frac{2}{\Om _0 \lf(s-2\ri)}\lf[
			(s-2) V - Q
			+\bigO{\eps ^2}
			\ri],
		\eeq
		where we set $ \gcrm : = g_{\betas} $ and
		\beq
			\label{eq: defining V Q}
			V:
			=\frac{\alpha ^2}{2}
			\int _{-\eta} ^\eta
			\diff y \: 
			y ^2 \gcrm^2,\qquad
			Q:
			=\frac{1}{2\pi}
			\int _{-\eta} ^\eta
			\diff y \: 
			\gcrm ^4.
		\eeq
	\end{pro}

	\begin{proof}
		The existence of a minimizer $ \betas $ is guaranteed from the fact that
		\[
			U _\beta (y)
			\ge \tx\frac{1}{\xinyp ^2}\lf[
			\frac{s-2}{2(s+2)}\eps ^4 \beta ^2
			-\Om _0 \eps ^4 \beta \eta
			\ri]
		\]
		which implies that $ \lim_{|\beta| \to \infty} \gvbe = + \infty $ (recall that $ s > 2 $).  By the same lower bound on the potential together with the trivial bound $ \gcre \le E _0 ^\mathrm{gv} = \OO(1) $, we also deduce that $ \betas = \bigO{\eps ^{-2}} $. 

		In order to find the explicit expression of $ \betas $, we first observe that by standard arguments $ \gvbe $ is a smooth function of $ \beta $ and therefore by the Feynman-Hellmann principle\footnote{The notation $ \braket{\cdot}{\cdot}_{\eta} $ stands for the scalar product in $ L^2([-\eta,\eta], (1 + \eps^2 y) \diff y) $.}
		\beq
			\partial _\beta \gvbe
			=\lf\langle
				\gvbm\lf|
				\partial _\beta U _\beta
			\ri|\gvbm
			\ri\rangle _\eta = \mean{\gvbm}{\tx\frac{\eps ^2}{\xinyp ^2}\lf(
				-2\Om _0 y
				-\Om _0 \eps ^2 y ^2
				+\eps ^2 \beta
			\ri)}{\gvbm}_{\eta}.
		\eeq
		Since $ \betas $ is a minimizer, we must have $ \lf. \partial_\beta \gvbe\ri|_{\betas} = 0 $, i.e.,
		\beq
			\label{eq: solved by bcr}
			\eps ^2 \betas
			\inpr{\gcrm}{\tx\frac{1}{\xinyp ^2}}{\gcrm}
			-2\Om _0
			\inpr{\gcrm}{\tx\frac{y}{\xinyp ^2}}{\gcrm}
			-\Om _0 \eps ^2
			\inpr{\gcrm}{\tx\frac{y ^2}{\xinyp ^2}}{\gcrm}
			=0.
		\eeq
		We compute the first and last terms of the expression above:
		\beq
			\label{eq: exp val betas}
			\inpr{\gcrm}{\tx\frac{1}{\xinyp ^2}}{\gcrm}
			=1
			+\bigO{\eps ^2},\qquad
			\inpr{\gcrm}{\tx\frac{1}{\xinyp ^2}y ^2}{\gcrm}
			= \frac{2V}{\alpha ^2}
			+\bigO{\eps ^2}.
		\eeq
		Indeed thanks to the exponential decay proven in \eqref{eq: estimate on gvbm}, one can easily realize that
		\beq
			\label{eq: int powers}
			\int_{-\eta}^{\eta} \diff y \: |y|^k \: \gcrm^2 = \OO(1),	\qquad		\mbox{for any } k < \infty.
		\eeq
		In fact an analogous estimate holds true if $ \gcrm $ is replaced with $ \gcrm^{\prime} $, in particular
		\beq
			\label{eq: int powers derivative}
			\int_{-\eta}^{\eta} \diff y \: y \: \lf( \gcrm'\ri)^2 = \OO(1).
		\eeq
		To see this it suffices to integrate by parts and use the variational equation \eqref{eq: gvbm variational} to go back to an expression involving only $ \gcrm $ and there one can use the above estimate. We omit the details for the sake of brevity.	Notice that at this stage we are implicitly exploiting the bound $ \betas = \OO(\eps^{-2}) $, which is among the hypothesis of Proposition \ref{pro: gvbm point est}. Next we integrate by parts the second term in \eqref{eq: solved by bcr} to get
		\bmln{
			\inpr{\gcrm}{\tx\frac{1}{\xinyp ^2}y}{\gcrm}
			=\int _{-\eta} ^\eta
			\diff y\
			\tx\frac{1}{\xiny} \: y \: \gcrm ^2 \\	
			= \lf[
			\tx\frac{1}{2\xinyp} \:y ^2 \: \gcrm ^2 
			\ri] _{-\eta} ^\eta
			- \disp\int _{-\eta} ^\eta
			\diff y\
			\tx\frac{1}{\xiny}\: y ^2 \: \gcrm \gcrm'
			+ \disp\frac{\eps ^2V}{\alpha ^2}
			\lf(
			1
			+\bigO{\eps ^2}
			\ri)	 
		}
		where the boundary terms (first term on the r.h.s. of the expression above) can be included in the remainder $ \OO(\eps^2) $ if we choose $ \eta _0 > 2 $ (see again \eqref{eq: estimate on gvbm}). For the rest we can compute
		\bmln{
			-\int _{-\eta} ^\eta
			\diff y \:
			\tx\frac{1}{\xiny} \: y ^2 \: \gcrm \gcrm ' = -
			\disp\frac{2}{\alpha ^2}\disp\int _{-\eta} ^\eta
			\diff y \: (1 + \eps^2 y) \:
			U _\betas(y)
			\gcrm 
			\gcrm'	\\
			- \frac{2}{\alpha ^2}
			\disp\int _{-\eta} ^\eta
			\diff y \: 
			\frac{
				2\Om _0
				\eps ^2 \betas y
				+\Om _0
				\eps ^4 \betas y ^2
				-\tx\frac{1}{2}\eps ^4 \betas ^2
			}{1 + \eps^2 y} \:
			\gcrm\gcrm'	\\
			=\frac{1}{\alpha ^2}
			\int _{-\eta} ^\eta
			\dey
			\partial_y \lf[
				-\tx\frac{1}{2} \lf(
				\gcrm'
				\ri) ^2
				+\frac{1}{2\pi} \gcrm ^4
				-\gcrchem \gcrm ^2
			\ri]
			+\frac{\eps ^2}{\alpha ^2}
			\int _{-\eta} ^\eta
			\dey
			y ^3
			v(y)
			\partial_y \gcrm ^2 	\\
			-\frac{\eps ^2}{\alpha ^2}
			\int _{-\eta} ^\eta
			\diff y \:
			\lf(
				\gcrm'
				\ri) ^2
			-\frac{\eps ^2 \betas}{\alpha ^2}
			\int _{-\eta} ^\eta \diff y \: \frac{
				2\Om _0
				y
				+\Om _0
				\eps ^2  y ^2
				-\tx\frac{1}{2}\eps ^2\betas 
			}{1 + \eps^2 y} \:
			\partial_y \gcrm^2=	\\
			=\frac{\eps ^2}{\alpha ^2}
			\int _{-\eta} ^\eta
			\diff y\
			\lf\{
				-\tx\frac{1}{2} \lf(
				\gcrm'
				\ri) ^2
				-\frac{1}{2\pi} \gcrm ^4
				+\gcrchem \gcrm ^2
			\ri\} 
			 - \frac{\eps^2(s+1) V}{\alpha^2}+ \frac{2 \eps^2 \Omega_0 \betas}{\alpha^2} + \OO(\eps ^4 \betas)
			+\bigO{\eps ^{4}}=	 	\\
			=\frac{\eps ^2}{\alpha^2}
			\lf[
				-K
				-Q
				+\gcrchem
				- (s+1) V + 2 \Omega_0 \betas +\bigO{\eps ^2 \betas}
				+\bigO{\eps ^2}
			\ri],
			}
			where we have made use repeatedly of \eqref{eq: int powers} and exploited the identity
			\bdm
				\lf( y^3 v(y) \ri)^{\prime} = \tx\frac{1}{2} \alpha^2 (s+1) y^2 + \OO(\eps^2 |y|^3).
			\edm
			We have also set
			\beq
				T :
			= \frac{1}{2}
		\int _{-\eta} ^\eta
		\diff y \: \lf(
			\gcrm'
		\ri)^2.
			\eeq
			Hence
			\beq
				\label{eq: implicit per bcr}
				\inpr{\gcrm}{\tx\frac{1}{\xinyp ^2}y}{\gcrm}
				=\frac{\eps ^2}{\alpha^2}
				\lf[
				-T - s V
				-Q
				+\gcrchem + 2 \Omega_0 \betas
				+\bigO{\eps ^2 \betas}
				+\bigO{\eps ^2}
				\ri]
			\eeq
			and plugging this together with \eqref{eq: exp val betas} into \eqref{eq: solved by bcr}, we obtain
			\[
				\frac{s-2}{s+2} \betas (1 + \OO(\eps^2))
				+ \frac{2\Om _0}{\alpha ^2}
				\lf[
				(s-2) V - Q
				\ri]
				+ \OO(\eps ^2)
				=0,
			\]
		since $ \gcrchem = T + V + 2Q + \OO(\eps^2) $ (see \eqref{eq: int powers} and \eqref{eq: int powers derivative}). The expression \eqref{eq: betas} is then recovered. 
	\end{proof}

\nt
Along the proof we have also proven in \eqref{eq: solved by bcr} that
\beq
	\label{eq: optimality}
	\int_{-\eta}^{\eta} \diff y \: \tx\frac{1}{1 + \eps^2 y} \lf(y + \frac{1}{2} \eps^2 y^2 - \frac{1}{2 \Omega_0} \eps^2 \betas \ri) \gcrm^2 = 0,
\eeq
which, thanks to the result about $ \betas $, also implies that
\beq
	\label{eq: gcrm almost symm}
	\inpr{\gcrm}{y}{\gcrm}
	= \OO(\eps^2),
\eeq
i.e., the profile $ \gcrm $ is almost symmetric w.r.t. the origin.
	
In fact this latter information can be deduced also by looking at the relation between the functional $ \gvf_{\betas} $ and its minimization and the limiting model $ \gvf $. From now on we fix $ \beta $ equal to the optimal value $ \betas $.

Before discussing this question further we have however to state an useful estimate on $ \gcrm $.

	\begin{lem}
		\mbox{}	\\
		There exists a finite constant $ C > 0 $ such that for any $ y \in [-\eta,\eta] $
		\beq \label{eq: derivative gcrm}
			\lf|
				\gcrm ' (y)
			\ri|
			\le C \eta ^3 \gcrm (y).
		\eeq
	\end{lem}

	\begin{proof}
		It suffices to integrate the variational equation \eqref{eq: gvbm variational} between $ y \geq y_{\betas} $ and $ \eta $ or $ -\eta $ and $ y \leq y_{\betas} $ (recall that $ y_{\betas} $ stands for the unique maximum point of $ \gcrm $): let us assume that $ y \geq \yb $, then by positivity of $ \gcrm $
		\bml{
			\tx\frac{1}{2} \lf|
			\gcrm'(y)
			\ri|
			=-\frac{1}{2} \gcrm' (y)
			=\disp\int _{y} ^{\eta}
			\diff t\:
			\lf\{ - \tx\frac{\eps^2}{2(1 + \eps^2 t)} \gcrm' + 
			\lf(\tx\frac{1}{\pi} \gcrm ^2
			+ U_{\betas}(t) + \eps^2 t^3 v(t)
			- \gvchem_{\betas}
			\ri) \gcrm \ri\}		\\
			\leq \disp\int _{y} ^{\eta}
			\diff t\:
			\lf\{ - \tx\frac{\eps^4}{2(1 + \eps^2 t)^2} + 
			\tx\frac{1}{\pi} \gcrm ^2
			+ U_{\betas}(t) + \eps^2 t^3 v(t)
			 \ri\}\gcrm +  \tx\frac{\eps^2}{2(1 + \eps^2 y)} \gcrm(y).
		}
		Now given that the quantity between brackets can be easily bounded from above by $ C \eta^2 $, it only remains to use the monotonicity of $ \gcrm $ to conclude the proof.
	\end{proof}
	
	\nt
	We are now in position to prove the first result about the energy difference $ \gcre - \gve $. As a matter of fact this will involve a corresponding statement about the closeness of $ \gcrm^2 $ to $ \gvm^2 $ in $ L^2 $. {We recall the expressions of the energy functionals
	\bdm
		\gvfs[g] = \int _{-\eta}^{\eta} \diff y \: (1 + \eps^2 y) \lf\{ \half
		\lf(
			g^{\prime}
		\ri)^2
		+U _{\betas}(y) g ^2 
		+\varepsilon ^2 y ^3 v(y) g ^2
		+\frac{1}{2\pi}
		g ^4
	\ri\},
	\edm
	\bdm
		\gvf[g] = \int _{\R} \diff y
		\lf\{
		\tx\frac{1}{2}\lf(
			g '
		\ri) ^2
		+ \frac{\alpha ^2}{2} y^2 g ^2 
		+\frac{1}{2\pi} g ^4
		\ri\}.
	\edm} 
	
	\begin{pro}[Estimate of $ \gcre - \gve $]
		\label{pro: gv energy diff}
		\mbox{}	\\
		As $ \eps \to 0 $
		\beq
			\label{eq: gv energy diff}
			\gcre
			=\gve
			+\bigO{\eps ^4 \eta ^7},	\qquad			\lf\|
			\gcrm ^2
			-\gvm ^2
			\ri\| _{L ^2 _\eta}^2
			=\bigO{\eps ^4 \eta ^{7}}.
		\eeq
	\end{pro}
	
	\begin{proof}
		We test the two functionals $ \gvf_{\betas} $ and $ \gvf $ on suitable test functions. Let us first regularize $ \gcrm $ outside $ [-\eta,\eta] $ to make it an admissible test function for $ \gvf $: we define
		\[
			g _\mathrm{trial}(y):
			= c_{\eps} 
			\begin{cases}
				\gvbm(y), & |y|\le\eta, \\
				r_{1}\lf(y\ri), & \eta\le y\le2\eta, \\
				r_{2}\lf(y\ri), & -2\eta\le y\le\eta, \\
				0, & |y|\ge2\eta,
			\end{cases}
		\]
		with $ r_{1,2} $ positive smooth functions chosen in such a way that $ \gtrial $ is at least $ C^2 $. We also assume that both functions $ r_{1,2} $ are also monotonically decreasing. The normalization constant $ c_{\eps}$, which ensures that $ \lf\| \gtrial \ri\| _{L^2(\R)} = 1 $, can be easily estimated: assuming that $ \eta_0 > 2 $, we have
		\beq
			c_{\eps} = 1 + \OO(\eps^4),
		\eeq
		since, e.g., 
		\bdm
			\int_{\eta}^{2\eta} \diff y \: r_1^2(y) \leq \eta \gcrm^2(\eta) = \OO(\eta\eps^{2\eta_0}).
		\edm 
		Notice also that we need to use \eqref{eq: gcrm almost symm} to reconstruct the norm of $ \gcrm $:
		\bdm
			\int_{-\eta}^{\eta} \diff y \: \gtrial^2 = c_{\eps}^2 \int_{-\eta}^{\eta} \diff y \: (1 + \eps^2 y) \: \gcrm^2 + c_{\eps}^2 \eps^2 \int_{-\eta}^{\eta} \diff y \: y \: \gcrm^2 = c_{\eps}^2 + \OO(\eps^4).
		\edm
		Now we estimate
		\[
			\gve
			\le\gvf[\gtrial]
			=c_{\eps} ^2
			\int _{-\eta} ^\eta
			\diff y\
			\lf\{
				\tx\frac{1}{2} (
				\gvbm '
			) ^2
			+\tx\frac{\alpha ^2}{2} y ^2 \gvbm ^2
			+\tx\frac{1}{2\pi} \gvbm ^4
			\ri\}
			+\OO(\eps^4).
		\]
		Thanks to \eqref{eq: gcrm almost symm} we can easily estimate the error we make by replacing $ \frac{\alpha ^2}{2} y ^2 $ with $ U _\beta + \eps ^2 y ^3 v(y) $: denoting for short $ \langle f \rangle : = \mean{\gcrm}{f}{\gcrm} $, we have
		\bmln{
			\lf\langle
			\gcrm\lf|
				\tx\frac{\alpha ^2}{2}y ^2
				-U_\beta(y)
				-\eps ^2 y ^3v(y)
			\ri|\gcrm
			\ri\rangle_{L^2(-\eta,\eta)}		\\
			=\lf\langle
			\gcrm\lf|
				\tx\frac{\alpha ^2 \eps ^2 y ^3\lf(
					2
					+\eps ^2 y
				\ri)}{2\xinyp ^2}
				+\frac{\eps ^2 \betas}{\xinyp ^2}\lf(
					2\Om _0 y
					+\Om _0 \eps ^2 y
					-\frac{1}{2}\eps ^2 \betas
				\ri)
				-\eps ^2 y ^3 v(y)
			\ri|\gcrm
			\ri\rangle			
			=\bigO{
			\eps ^4
			+\eps ^2 \lf\langle	y ^3 \ri\rangle 
		},}
		so that
		\beq
			\label{eq: gve ub}
			\gve
			\le 
			\gcre
			+\bigO{
			\eps ^4
			+\eps ^2 \lf\langle	y ^3 \ri\rangle 
		},
		\eeq
		where we have also used \eqref{eq: en preliminary est}, which in turn requires $ \betas = \OO(\eps^{-2}) $.
	
		The trial state for the functional $ \gvf_{\betas}$ is simply the truncation of $ \gvm $, i.e., $ c_{\eps} \gvm $, where now the normalization factor can be estimated in this case as
		\beq
			c_{\eps} = 1 + \OO(\eps^\infty),
		\eeq
		since by symmetry of $ \gvm $
		\bdm
			\int_{-\eta}^{\eta} \diff y \: (1 + \eps^2 y) \: \gvm^2 = 1 - 2 \int_{\eta}^{\infty} \diff y \: \gvm^2 = 1 + \OO(\eps^{\infty}),
		\edm 
		where we have used the pointwise estimate \eqref{eq: estimate on gvm} on $ \gvm $ and the fact that the integral of a gaussian, i.e., the error function, is bounded by the value of the gaussian at the boundary, i.e., $ \exp\{- c |\log\eps|^2 \} = \OO(\eps^{\infty}) $ (see, e.g., \cite[Eq. (5.1.19)]{AS}). Then we have
		\bml{
			\label{eq: gcre ub}
			\gcre
			\le \gvf_{\betas}[c_{\eps} \gvm]
			= (1 + \OO(\eps^{\infty})) \int _{-\eta} ^\eta
			\dey
			\lf\{
			\tx\frac{1}{2} (\gvm^{\prime})^2+U _\betas(y) \gvm^2
			+\eps ^2 y ^3 v(y) \gvm^2
			+\frac{1}{2\pi}\gvm^4
			\ri\}
			\\
			=
			\int _{-\eta} ^\eta
			\diff y 
			\lf\{
			\tx\frac{1}{2} (\gvm^{\prime})^2+ \tx\frac{\alpha^2}{2} \gvm^2
			+\frac{1}{2\pi}\gvm^4
			\ri\}
			+ \OO(\eps^4)
			= \gve
			+ \OO(\eps^4).
		}
		Putting together \eqref{eq: gve ub} with \eqref{eq: gcre ub}, we obtain
		\beq
			\label{eq: en diff est 1}
			\gcre
			= \gve
			+\bigO{
			\eps ^4
			+\eps ^2 \lf\langle	y ^3 \ri\rangle 
			}.
		\eeq
		
		Now we decouple the energy $ \gve $: first we bound from below $ \gve $ as
		\bdm
			\gve \geq \int_{-\eta}^{\eta} \diff y \: (1 + \eps^2 y) \lf\{ \tx\frac{1}{2} (\gvm^{\prime})^2+ \tx\frac{\alpha^2}{2} y^2 \gvm^2
			+\frac{1}{2\pi}\gvm^4 \ri\},
		\edm
		where we have just dropped some positive quantities and used the symmetry of $ \gvm $. Then we set $ \gvm = u \gcrm $ for some unknown smooth function $ u $ (recall that $ \gcrm $ never vanishes in $ [-\eta,\eta] $) and using the variational equation for $ \gcrm $ as well as Neumann boundary conditions, we obtain
		\bmln{
			\gve \geq \gcre + \int_{-\eta}^{\eta} \diff y \: (1 + \eps^2 y) \: \gcrm^2 \lf\{\tx\frac{1}{2} (u')^2 + \lf( \frac{\alpha^2}{2} y^2 - U_{\betas}(y) - \eps^2 y^3 v(y) \ri) u^2 + \frac{1}{2\pi} \gcrm^2 \lf( 1 - u^2\ri)^2 \ri\} 	\\
		- \tx\frac{1}{2}\eps^2 \disp\int_{-\eta}^{\eta} \diff y \: u^2 \gcrm \gcrm' + \OO(\eps^{\infty}).
		}
		Then we estimate
		\bmln{
			 \lf|\int_{-\eta}^{\eta} \diff y \: (1 + \eps^2 y) \: \lf( \tx\frac{\alpha^2}{2} y^2 - U_{\betas}(y) - \eps^2 y^3 v(y) \ri) \gcrm^2 u^2 \ri| \leq C \eps^2 \int_{-\eta}^{\eta} \diff y \: |y| \gcrm^2 \lf|1 - u^2 \ri| \\
			 +	 \lf| \int_{-\eta}^{\eta} \diff y \: (1 + \eps^2 y) \: \lf( \tx\frac{\alpha^2}{2} y^2 - U_{\betas}(y) - \eps^2 y^3 v(y) \ri) \gcrm^2 \ri| \\
			 \leq C \eps^2 \eta^{3/2} \lf\| \gcrm^2(1 - u^2) \ri\|_{L^2_{\eta}} +  \OO\lf(\eps^2 \lf\langle y^3 \ri\rangle +\eps^4\ri),
		}
		and by \eqref{eq: derivative gcrm} (notice that the factor $1 + \eps^2 y $ is uniformly bounded from above and below by a constant)
		\bmln{
			\lf| \disp\int_{-\eta}^{\eta} \diff y \: u^2 \gcrm \gcrm' \ri| = \disp\int_{-\eta}^{\eta} \diff y \: \lf|1 - u^2 \ri| |\gcrm| |\gcrm'| + \OO(\eps^4) \leq C \eta^3 \disp\int_{-\eta}^{\eta} \diff y \: \lf|1 - u^2 \ri| \gcrm^2 + \OO(\eps^4)	\\
			 \leq C \eta^{7/2} \lf\| \gcrm^2(1 - u^2) \ri\|_{L^2_{\eta}} + \OO(\eps^4),
		}
		which imply by dropping the kinetic term
		\bml{
			\label{eq: en diff est 2}
			\gve \geq \gcre + \tx\frac{1}{2} \lf\| \gcrm u^{\prime} \ri\|_{L^2_{\eta}}^2 + \tx\frac{1}{\pi} \lf\| \gcrm^2(1 - u^2) \ri\|^2_{L^2_{\eta}}- C \eps^2 \eta^{7/2} \lf\| \gcrm^2(1 - u^2) \ri\|_{L^2_{\eta}} +  \OO\lf(\eps^2 \lf\langle y^3 \ri\rangle +\eps^4\ri)	\\
			\geq \gcre + \tx\frac{1}{\pi} \lf(\lf\| \gcrm^2(1 - u^2) \ri\|_{L^2_{\eta}}- C \eps^2 \eta^{7/2} \ri)^2 +  \OO\lf(\eps^2 \lf\langle y^3 \ri\rangle +\eps^4 \eta^7\ri).
		}
		If we compare what we have obtained with \eqref{eq: en diff est 1}, we conclude that
		\beq
			\lf\| \gcrm^2 - \gvm^2\ri\|^2_{L^2_{\eta}} = \lf\| \gcrm^2(1 - u^2) \ri\|_{L^2_{\eta}}^2 =  \OO\lf(\eps^2 \lf\langle y^3 \ri\rangle +\eps^4 \eta^7\ri),
		\eeq
		but on the other hand
		\beq
			\label{eq: y3 proof}
			\lf| \int_{-\eta}^{\eta} \diff y \: y^3 \: \gcrm^2 \ri| = \lf| \int_{-\eta}^{\eta} \diff y \: y^3 \: \lf( \gcrm^2 - \gvm^2 \ri) \ri| \leq C \eta^{7/2} \lf\| \gcrm^2 - \gvm^2\ri\|_{L^2_{\eta}},
		\eeq
		so that finally
		\bdm
			\lf\| \gcrm^2 - \gvm^2\ri\|^2_{L^2_{\eta}} = \OO\lf(\eps^4 \eta^7\ri).
		\edm
		This proves the second inequality in \eqref{eq: gv energy diff} but the first is obtained by replacing the above estimate into \eqref{eq: en diff est 1}.
	\end{proof}
	
	\nt
	It is interesting to remark that a by-product of the proof of Proposition \ref{pro: gv energy diff} is that (see \eqref{eq: y3 proof})
	\beq
		\label{eq: y3}
		\mean{\gcrm}{y^3}{\gcrm} = \OO(\eps^2 \eta^7),
	\eeq
	which in combination with \eqref{eq: gcrm almost symm} is a strong indication of $ \gcrm $ being symmetric w.r.t. the origin with a very high precision. In fact this is also made apparent by the estimate of the difference $ \gcrm - \gvm $.

	The $L^2$-statement in \eqref{eq: gv energy diff} can indeed be improved to an $ L^{\infty}$-one, showing that $ \gcrm $ and $ \gvm $ are pointwise close. The price to pay to have a result in a stronger norm is the restriction of the region under consideration to the annulus $ \annt \subset \ann $ defined as 
	\beq
		\annt :
		=\lf\{
			y \in \mathbb{R}\;\big|\;
			\gvm(y) \geq
			 \tx\frac{1}{\eta ^\nu}
		\ri\},
	\eeq	
	for some $ \nu > 0 $ independent of $ \eps $. Note that thanks to the pointwise estimates \eqref{eq: estimate on gvm} and the monotonicity of $ \gvm $, 
	\beq
		\annt = [-y_{\eta},y_{\eta}],	\qquad		\mbox{with } y_{\eta} \gg 1.
	\eeq
	
	\begin{pro}[Pointwise estimate of $ \gcrm - \gvm $]
		\label{pro: infty est diff gvm}	
		\mbox{}	\\
		As $ \eps \to 0 $ and for any $ \nu > 0 $
		\beq
			\label{eq: infty est diff gvm}
			\lf\|
				\gcrm
				-\gvm
				\ri\| _{L ^\infty \lf(\annt\ri)}
			=\bigO{\eps ^2 \eta ^{\frac{7+4\nu}{2}}}.
		\eeq
	\end{pro}

	\begin{proof}
		Going back to \eqref{eq: en diff est 2} and retaining the kinetic term, we see that we obtain via \eqref{eq: gv energy diff} the upper bound
		\beq
			\lf\| \gcrm u^{\prime} \ri\|_{L^2_{\eta}}^2 = \OO(\eps^4 \eta^7),
		\eeq
		where we recall that $ u = \gvm/\gcrm $. Now let us introduce the set
		\bdm
			\annh : = \lf\{
			y \in \mathbb{R}\;\big|\;
			\gcrm(y) \geq
			 \tx\frac{1}{2\eta ^\nu}
		\ri\},
		\edm
		so that above inequality  together with \eqref{eq: gv energy diff} and the bound $ |1 - u| \leq |1 - u^2| $ imply
		\bdm
			\lf\| u^{\prime} \ri\|_{L^2(\annh)}^2 = \OO(\eps^4 \eta^{7+2\nu}),	\qquad \lf\| 1 - u \ri\|^2_{L^2(\annh)} = \OO(\eps^4 \eta^{7+4\nu}).
		\edm
		Then it suffices to use Sobolev inequality in one-dimension: 
		\beq
			\lf\| 1 - u \ri\|_{L^{\infty}(\annh)}^2 \leq C \lf( \lf\| u' \ri\|_{L^{2}(\annh)}^2 + \lf\| 1 - u \ri\|_{L^{2}(\annh)}^2 \ri) = \OO\lf(\eps^4 \eta^{7+4\nu} \ri). 
		\eeq
		Finally to obtain the result it remains to observe that $ \annt \subset \annh $, because in the region where $ \gcrm \geq 1/(2\eta^\nu) $, by the pointwise estimate, $ \gvm $ is larger than $ (1 + o(1))/(2\eta^\nu) $, which is obviously satisfied if $ \gvm \geq 1/\eta^\nu  $.
	\end{proof}
	
	\nt 
	The above bound shows that inside $ \annt $ one can estimate the distance of $ \gcrm $ from a perfectly even function: for any $ y \in \annt $
	\bdm
		\gcrm(-y)
		= \gcrm(y)
		+\bigO{\eps ^2 \eta ^{\frac{7+4\nu}{2}}},
	\edm
	which is perfectly compatible with the estimates \eqref{eq: gcrm almost symm} and \eqref{eq: y3}.

	Another useful consequence of the above pointwise statement is the following
	
	\begin{cor}[Maximum point of $ \gcrm $]
		\label{cor: max gcrm}
		\mbox{}	\\
		Let $ y_{\betas} $ be the unique maximum point of $ \gcrm $, then as $ \eps \to 0 $
		\beq
			y_{\betas} = \bigO{\eps ^2 \eta ^{\frac{7+4\nu}{2}}}.
		\eeq
	\end{cor}
	
	\begin{proof}
		The result is a straightforward consequence of the pointwise estimate \eqref{eq: infty est diff gvm} and the properties of $ \gvm $ (see Proposition \ref{pro: min gvf}).	
	\end{proof}

\subsection{Critical Velocity and Positivity of the Cost Function}
\label{sec: positivity}

From now we fix the phase to be optimal one, i.e., $ \beta = \betas $. The {\it potential function} is defined as
\bml{
	\label{eq: potential f}
	F(y):
	=- \frac{1}{\eps^2} \int _{-\eta} ^y
	\diff t \: (1 + \eps^2 t)
	\lf. \partial _\beta U _\beta (t) \ri|_{\beta = \betas}
	\gcrm ^2	\\
	=2 \Omega_0 \int_{-\eta}^{y} \diff t \: \frac{1}{1 + \eps^2 t} \lf(t + \frac{1}{2} \eps^2 t^2 - \frac{\eps^2 \betas}{2 \Omega_0}  \ri) \gcrm^2.
}
The main object under investigation is the {\it cost function}
\beq
	\label{eq: cost f}
	K(y) : = \tx\frac{1}{2} \gcrm^2(y) + F(y),
\eeq
and our main goal in this Section is to prove that it is positive in the bulk of the condensate when $ \Omega_0 \geq \Omegac $. To this purpose we will clearly have to investigate the equation \eqref{eq: Omegac} and prove at least that there exists a positive solution to it. Notice the equation \eqref{eq: Omegac} involves only quantities relative to the limiting functional $ \gvf $ and is independent of $ \eps $. We thus introduce the analogue of \eqref{eq: cost f} for the limiting case, i.e., the function \eqref{eq: kgv}
\bdm
	\kgv = \half \gvm^2 + \fgv,
\edm
where $ \fgv $ is defined in \eqref{eq: fgv}:
\bdm
	\fgv(y) = - 2 \Omega_0 \int_y^{\infty} \diff t \: t \: \gvm^2(t).
\edm

We will start by studying the positivity of \eqref{eq: kgv} and show that the condition $ \Omega_0 > \Omegac $, where the latter is defined as the biggest solution to \eqref{eq: Omegac}, is sufficient to deduce that $ \kgv(y) \geq 0 $ for any $ y \in \R $. In the second part of the Section we will turn our attention to the cost function \eqref{eq: cost f} and prove that the same condition on $ \Omega_0 $ guarantees positivity of $ K $ as well.

We first observe that $ \fgv $ is a negative function vanishing at $ \pm \infty $: at $ y = +\infty $ it is obvious, at $ -\infty $ it is a consequence of parity of $ \gvm $. By this property one can rewrite
\bdm
	\fgv(y) = 2 \Omega_0 \int_{-\infty}^y \diff t\: t \: \gvm^2(t).
\edm
In fact there is another explicit expression of $ \fgv $, which can be obtained by using the variational equation \eqref{eq: gvm variational}:
\bml{
	\label{eq: fgv alternative}
	\fgv(y) =  - \Omega_0 \int_y^{\infty} \diff t \: \partial_t (t^2) \: \gvm^2 = \Omega_0 y^2 \gvm^2(y) + 2 \Omega_0 \int_y^{\infty} \diff t \: t^2 \: \gvm \gvm^{\prime}	\\
	 = \Omega_0 y^2 \gvm^2(y) + \frac{4 \Omega_0}{\alpha^2} \int_y^{\infty} \diff t \: \gvm^{\prime} \lf[ \tx\frac{1}{2} \gvm^{\prime\prime} - \frac{1}{\pi} \gvm^3 + \gvchem \gvm \ri] 	\\
	 =
	- \tx\frac{1}{\Omega_0(s+2)} \lf( \gvm^{\prime}(y) \ri)^2 + \lf[ \Omega_0 y^2 + \frac{1}{\pi \Omega_0 (s+2)} \gvm^2(y) - \frac{2\gvchem}{\Omega_0(s+2)} \ri] \gvm^2(y),
	}
where we have used the exponential decay at $ \infty $ of $ \gvm $ to cancel the missing boundary terms. Consequently we can rewrite $ \kgv $ as
\beq
	\label{eq: kgv alternative}
	\kgv(y) = - \tx\frac{1}{\Omega_0(s+2)} \lf( \gvm^{\prime}(y) \ri)^2 + \lf[ \tx\frac{1}{2} + \Omega_0 y^2 + \frac{1}{\pi \Omega_0 (s+2)} \gvm^2(y) - \frac{2\gvchem}{\Omega_0(s+2)} \ri] \gvm^2(y).
\eeq
The main result about $ \kgv $ is the following

	\begin{pro}[Positivity of $ \kgv $]
		\label{pro: kgv positive}
		\mbox{}	\\
		Let $ \Omega_0 > 0 $, then
		\beq
			\label{eq: kgv positive}
			\kgv(y) \geq 0 \mbox{ for any } y \in \R 	\Longleftrightarrow \Omega_0 \geq \tx\frac{4}{s+2} \lf[ \gvchem - \frac{1}{2\pi} \gvm^2(0) \ri].
		\eeq
		Moreover if the strict inequality is verified on the r.h.s., $ \kgv(y) > 0 $ for any $ y $ finite.
	\end{pro}
	
	{
	\begin{rem}[Comparison with \cite{CPRY3}]
		\label{rem: cost functions}
		\mbox{}	\\
		Despite the use of two different potential functions $ F_1 $ and $ F_2 $ in \cite{CPRY3}, one should realize that \cite[Lemma 3.3]{CPRY3} yields the poitwise positivity of a cost function, which is the analogue of $ \kgv $ \underline{in the asymptotic regime} $ \Omega_0 \gg 1 $. In fact it can be easily seen that the cost function in \cite{CPRY3} is bounded from below by $ \kgv $ and therefore the positivity of the latter implies the positivity of the first. Hence any threshold $ \bar\Omega_0 $ one might deduce there must be larger than $ \Omegac $ by definition.
	\end{rem}}
	
	\begin{proof}
		One side of the statement, i.e., the fact that the condition $ \kgv(0) \geq 0 $ is necessary for the positivity of $ \kgv $ everywhere, is obviously trivial, so we focus on the other side of the implication, namely that $ \kgv(0) \geq 0 $ is also sufficient.
		
		 The core of the proof is to show that
		 \beq
		 	\label{eq: condition kgv}
		 	\kgv(y) \geq 0 \mbox{ for any } y \in \R 	\Longleftrightarrow \kgv(0) \geq 0.
		\eeq
		Indeed if we assume that this double implication is true, a straightforward computation yields
		\beq
			\label{eq: kgv0}
			\kgv(0) = \lf[ \frac{1}{2} + \frac{1}{\pi \Omega_0 (s+2)} \gvm^2(0) - \frac{2\gvchem}{\Omega_0(s+2)} \ri] \gvm^2(0),
		\eeq
		since $ \gvm $ is symmetric w.r.t. the origin and has a maximum at $ y = 0  $ (see Proposition \ref{pro: min gvf}). The result is then a trivial consequence of strict positivity of $ \gvm(0) $.
		
		In order to prove \eqref{eq: condition kgv}, we first observe that $ \kgv(\pm \infty) = 0 $ and $ \kgv $ is smooth, so, if there was a point $ y_0 $ where $ \kgv $ becomes negative, it must be $ |y_0| < +\infty $. Moreover as $ \gvm $, $ \kgv $ is symmetric w.r.t. to the origin, so it suffices to consider $ y \in \R^+ $. The derivative of $ \kgv $ is easily computed from the expression \eqref{eq: kgv}:
		\beq
			{\kgv}^{\prime}(y) = \gvm \gvm^{\prime} + 2 \Omega_0 y \gvm^2,
		\eeq
		and one immediately has that $ {\kgv}^{\prime}(0) = 0 $, i.e., $ \kgv $ has a critical point there. Whether it is a minimum or a maximum depends on $ s $ and $ \Omega_0 $, but as we are going to see this does not matter. We can in any case compute easily the second derivative of $ \kgv $ exploiting once more \eqref{eq: gvm variational}:
		\beq
			\label{eq: kgv second derivative}
			{\kgv}^{\prime\prime}(y) = \lf(\gvm^{\prime}\ri)^2 + 4 \Omega_0 y \gvm \gvm^{\prime} + 2 \lf[ \tx\frac{1}{2} \alpha^2 y^2 + \Omega_0  - \gvchem+ \tx\frac{1}{\pi} \gvm^2 \ri] \gvm^2.
		\eeq

		Then we prove the crucial property of $ \kgv $: suppose that $ \kgv $ has a maximum at $ y_1  \geq 0 $ and then a minimum at $ y_2 > y_1 $, then
		\beq
			\label{eq: kgv crucial pro}
			\frac{\kgv(y_2)}{\gvm^2(y_2)} \geq \frac{\kgv(y_1)}{\gvm^2(y_1)},
		\eeq
		and in particular $ \kgv(y_2) \geq 0 $ if $ \kgv(y_1) \geq 0 $. 
		
		To conclude the argument once \eqref{eq: kgv crucial pro} is proven, it is sufficient to observe that $ \kgv $ has a critical point in $ y = 0 $, which by parity must be either a maximum or a minimum: if it is a maximum, then \eqref{eq: kgv crucial pro} shows that at any minimum point $ y_2 > 0 $, $ \kgv(y_2) \geq 0 $. Notice that it does not matter whether $ \kgv $ has a single or multiple minima, because any minimum after the first requires the presence of a preceding maximum point, where $ \kgv $ is larger than its first minimum and therefore positive. If on the opposite $ \kgv $ has a minimum at the origin, then it means that there must be a maximum at some $ y_1 > 0 $, where obviously $ \kgv(y_1) \geq \kgv(0) \geq 0 $ and we can repeat the argument for any minimum after $ y_1 $.
		
		Let us now prove \eqref{eq: kgv crucial pro}: we assume again that $ \kgv $ has a maximum in $ y_1 \geq 0 $ and a minimum in $ y_2 > y_1 $. Then it must be $ {\kgv}^{\prime}(y_{1,2}) = 0 $, i.e.,
		\beq
			\gvm^{\prime}(y_{1,2}) = - 2 \Omega_0 y_{1,2} \gvm^2(y_{1,2}).
		\eeq
		Moreover replacing this condition in \eqref{eq: kgv} and \eqref{eq: kgv second derivative}, we get
		\beq
			\label{eq: kgv min}
			\kgv(y_{1,2}) = \lf[ \Omega_0 \tx\frac{s-2}{s+2} y_{1,2}^2 + \frac{1}{2}  - \frac{2}{\Omega_0(s+2)} \lf( \gvchem - \frac{1}{2\pi} \gvm^2(y_{1,2}) \ri) \ri] \gvm^2(y_{1,2}),
		\eeq
		\beq
			\label{eq: kgv second derivative min}
			{\kgv}^{\prime\prime}(y_{1,2}) = \lf[ \Omega_0^2 (s-2) y_{1,2}^2 + 2\Omega_0 - 2 \gvchem + \tx\frac{2}{\pi} \gvm^2(y_{1,2}) \ri] \gvm^2(y_{1,2}).
		\eeq
		Moreover
		\beq
			{\kgv}^{\prime\prime}(y_1) \leq 0 \leq {\kgv}^{\prime\prime}(y_2),
		\eeq
		which implies
		\beq
			\Omega_0^2(s-2) y_1^2 + 2 \Omega_0 - 2 \mu + \tx\frac{2}{\pi} \gvm^2(y_{1}) \leq \Omega_0^2(s-2) y_2^2 + 2 \Omega_0 - 2 \mu + \tx\frac{2}{\pi} \gvm^2(y_{2}),
		\eeq
		and using this inequality in the expression of $ \kgv(y_2) $, we obtain
		\bml{
			\frac{\kgv(y_2)}{\gvm^2(y_2)} \geq \Omega_0 \frac{s-2}{s+2} y_{1}^2 - \frac{2}{\pi\Omega_0(s+2)} \lf( \gvm^2(y_2) - \gvm^2(y_1) \ri) + \frac{1}{2}  - \frac{2}{\Omega_0(s+2)} \lf( \gvchem - \frac{1}{2\pi} \gvm^2(y_{2}) \ri) \\
			= \frac{\kgv(y_1)}{\gvm^2(y_1)} +  \frac{1}{\pi\Omega_0(s+2)} \lf( \gvm^2(y_1) - \gvm^2(y_2) \ri) \geq \frac{\kgv(y_1)}{\gvm^2(y_1)},
		}
		because by hypothesis $ y_2 > y_1 $ and $ \gvm $ is decreasing.
		
		Notice that as a by-product of our analysis we found out that can have no global minima, since $ \inf_{y \in \R} \kgv(y) = 0 $ and therefore at any such minimum point $ y_0 $, we would have $ \kgv(y_0) = 0 $, but this clearly contradicts \eqref{eq: kgv crucial pro}. Hence if the inequality on r.h.s. of \eqref{eq: kgv positive} is strict then $ \kgv(y) > 0 $ for any finite $ y $.
	\end{proof}
	
\nt
Proposition \ref{pro: kgv positive} introduces the equation \eqref{eq: Omegac}. The next step is obviously to prove that such an equation as at least one solution:

	\begin{pro}[Equation \eqref{eq: Omegac}]
		\label{pro: Omegac}
		\mbox{}	\\
		The equation \eqref{eq: Omegac}
		\bdm
			 \Omega_0 = \frac{4}{s+2} \lf[ \gvchem - \frac{1}{2\pi} \gvm^2(0) \ri].
		\edm
		has at least one solution $ \Omega_0 > 0 $.
	\end{pro}
	
	\begin{proof}
		Let us first set 
		\bdm
			G(\Omega_0) : = \frac{4}{s+2} \lf[ \gvchem - \frac{1}{2\pi} \gvm^2(0) \ri],
		\edm
		so that \eqref{eq: Omegac} reads $ \Omega_0 = G(\Omega_0) $.	We will show that $ G(\Omega_0) $ is asymptotically smaller than $ \Omega_0 $ (resp. larger) $ \Omega_0 $ for large (resp. small) $ \Omega_0 $.
		
		Let us first consider $ \Omega_0 \gg 1 $: using the trivial bound $ \lf\| \gvm \ri\|_4^4 \leq \lf\| \gvm \ri\|_{\infty}^2 $ and the definition of $ \gvchem $, we get 
		\bdm
			G(\Omega_0) \leq \tx\frac{4}{s+2} \gve.
		\edm
		If now we plug into $ \gvf $ as a trial state the ground state of the harmonic oscillator $ \hosc = - \frac{1}{2} \Delta + \frac{1}{2} \alpha^2 y^2 $, we easily obtain
		\bdm
			\gve \leq \tx\frac{1}{2} \Omega_0 \sqrt{s+2} \lf(1 + \OO(\Omega_0^{-1/2}) \ri),
		\edm
		so that
		\bdm
			G(\Omega_0) \leq \tx\frac{2}{\sqrt{s+2}} \Omega_0  \lf(1 + \OO(\Omega_0^{-1/2}) \ri) < \Omega_0,
		\edm
		if $ \Omega_0 \gg 1 $, because $ s > 2 $.
		
		On the other hand for small $ \Omega_0 $, thanks to the estimate \eqref{eq: gvm infty bound} and again the definition of $ \gvchem $, we have
		\bdm
			G(\Omega_0) \geq \tx\frac{2}{s+2} \gve.
		\edm
		To bound from below $ \gve $ for small $ \Omega_0 $, we can simply drop the kinetic term to get
		\bdm	
			\gve \geq \inf_{\lf\| \rho \ri\|_1 = 1} \int_{\R} \diff y \: \lf\{ \tx\frac{1}{2} \alpha^2 y^2 \rho + \frac{1}{2\pi} \rho^2 \ri\},
		\edm
		i.e., a TF-like functional. By scaling we immediately obtain
		\bdm
			\inf_{\lf\| \rho \ri\|_1 = 1} \int_{\R} \diff y \: \lf\{ \tx\frac{1}{2} \alpha^2 y^2 \rho + \frac{1}{2\pi} \rho^2 \ri\} = C \Omega_0^{2/3},
		\edm
		as $ \Omega_0 \to 0 $, so that 
		\bdm
			G(\Omega_0) \geq C \Omega_0^{2/3} > \Omega_0,
		\edm
		for $ \Omega_0 $ small enough.
	\end{proof}
	
	In the above Proposition we have not investigated the uniqueness of the solution. We indeed expect that such a solution is in fact unique, but without a proof of this fact, we have to choose $ \Omegac $ equal to the {\it largest} possible solution.

	Now we turn our attention back to the cost function $ K $: as $ \kgv $ it is given by the sum of a the positive density $ \frac{1}{2} \gcrm^2 $ and the negative potential function $ F $ (see next Proposition \ref{pro: pro F}). In addition $ \gcrm $ is monotonically decreasing for $ y \geq y_{\betas} = o(1) $ and, like $ F $, almost symmetric. Close the boundary of the interval $ [-\eta,\eta] $, $ \gcrm $ gets extremely small (in fact exponentially small in $ \eps $) but $ F $ vanishes identically at $ \pm \eta $. In conclusion it is clear that the overall positivity of $ K $ should then emerge from a very delicate balance between the two opposite contributions.

	We first state some simple properties of $ F $ collected in the following

	\begin{pro}[Properties of $ F $]
		\label{pro: pro F}
		\mbox{}	\\
		The potential function defined in \eqref{eq: potential f} is such that
		\beq 
			\label{eq: negativity F}
			F(y) \leq 0,	\qquad 	\mbox{for any } y \in [-\eta,\eta],
		\eeq
		\beq
			\label{eq: boundary F}
			F(\pm \eta) = 0.
		\eeq	
	\end{pro}
	
	\begin{proof}
		One of the identities \eqref{eq: boundary F} is trivial, the other is a direct consequence of \eqref{eq: optimality}. In order to show that $ F $ is negative everywhere we compute the derivative
		\beq
			\label{eq: derivative F}
			F'(y)
			= \tx\frac{1}{\xiny}
			\lf(
			2\Om _0 y
			+\Om _0 \eps ^2 y ^2
			- \betas \eps ^2
			\ri)\gcrm ^2(y),
		\eeq
		and it is easy to verify that because of the first term $ F'(-\eta) < 0 $ while $ F'(\eta) > 0 $. Moreover $ F' $ vanishes at a single point $ y_F = \OO(\eps^2) $, where $ F $ has a global minimum. Hence it is negative everywhere in $ [-\eta,\eta] $.
	\end{proof}
	
	A very crucial piece of information about the potential function formulated in the next Proposition is an alternative expression of it, which relies on the variational equation \eqref{eq: gvbm variational} and is the analogue of \eqref{eq: fgv alternative} for $ \fgv $.
	
	\begin{pro}[Alternative expressions of $ F $]
		\label{pro: alt expression F}
		\mbox{}	\\
		For any $ y \in [-\eta,\eta] $ the potential function $ F $ admits the following alternative expressions
		\beq
			F(y)
			= - \frac{\Om _0}{\alpha ^2} \lf(
				\gcrm' 
			\ri) ^2 
			+ \frac{2\Om _0}{\alpha ^2}\lf[ \frac{1}{2} \alpha^2 y^2 
			+\frac{1}{2\pi} \gcrm^2
			-\gcrchem 
			\ri] \gcrm ^2 +
			\begin{cases}
				R_{+}(y) + R_{+},	&	\mbox{if } y \geq y_\betas,	\\
				R_{-}(y) + R_{-},	&	\mbox{if } y \leq y_\betas,
			\end{cases}
		\eeq
		where 
		\beq
			\label{eq: tails}
			R_{\pm}(y) = \OO(\eps ^2 \eta ^7) \gcrm^2(y), 	\qquad
			R_{\pm} = -\Om _0 \eta ^2 \gcrm ^2(\pm\eta)
			\lf(
			1
			+ o(1) \ri).
		\eeq
	\end{pro}

	\begin{proof}
		We consider only the case $ y \ge \yb $, since the other one is analogous. The key ingredient of the proof is an integration by parts, exactly as for \eqref{eq: fgv alternative}. We spell all the details nevertheless for the sake of clarity. Thanks to the vanishing of $ F $ at $ \eta $, we have
		\bml{
			\label{eq: F computed 0}
			F(y)
			=\int _y ^\eta
			\diff t\
			\tx\frac{1}{\xint}
			\lf(
			-2\Om _0 t
			-\Om _0 \eps ^2 t ^2
			+\betas \eps ^2
			\ri)
			\gcrm ^2 = \disp\int _y ^\eta
			\diff t\
			\tx\frac{1}{\xint}
			\: \gcrm ^2	\: \partial_t \lf(
			-\Omega_0 t ^2
			-\frac{1}{3} \Om _0 \eps ^2 t ^3
			+\betas \eps ^2 t
			\ri)
			\\
			= \tx\frac{1}{1+\eps ^2 \eta}
			\lf(
			-\Om _0\eta ^2
			-\frac{1}{3} \Om _0 \eps ^2 \eta ^3
			+\betas \eps ^2 \eta
			\ri)
			\gcrm ^2 (\eta)
			-\frac{1}{\xiny}
			\lf(
			-\Om _0y ^2
			-\frac{1}{3}\Om _0 \eps ^2 y ^3
			+\betas \eps ^2 y
			\ri)
			\gcrm ^2 (y)+
			\\
			+\int _y ^\eta
			\diff t\	
			\lf(
			-\Om _0 t ^2
			-\tx\frac{1}{3}\Om _0 \eps ^2 t ^3
			+\betas \eps ^2 t
			\ri)
			\lf[ \tx\frac{\eps ^2}{\xinyp ^2} \gcrm ^2 - \frac{2}{\xinyp} \gcrm  \gcrm' \ri]  \\
			= 
			-\Om _0 \eta ^2 \gcrm ^2 (\eta) (1 + o(1))
			+ \lf( \Om _0 y ^2 + \OO(\eps^2\eta^4) \ri) \gcrm^2(y) 	\\
			- 2 \disp\int _y ^\eta
			\diff t\: \tx\frac{1}{\xinyp} 	
			\lf(
			-\Om _0 t ^2
			-\tx\frac{1}{3}\Om _0 \eps ^2 t ^3
			+\betas \eps ^2 t
			\ri)
			\gcrm  \gcrm'.
		}	
		We now rewrite the last term by reconstructing the potential $ U_{\betas} $ and using the variational equation \eqref{eq: gvbm variational}: since 
		\bmln{
			-\tx\frac{1}{\xiny}
			\lf(
			-\Om _0 t ^2
			-\frac{1}{3} \Om _0\eps ^2 t ^3
			+\betas \eps ^2 t
			\ri) = \frac{2\Om _0}{\alpha ^2} (1 + \eps^2 t) U _\betas(t) \\
			+ \tx\frac{\eps ^2}{\xintp}
			\lf(
			\frac{1}{6} \alpha^2 t ^3
			- \frac{1}{2} \Omega_0 (s-2) \betas t 
			+\Om _0 \eps ^2 \betas t ^2-\frac{1}{2}\eps ^2 \betas^2 
			\ri),
		}
		we obtain
		\bml{
			\label{eq: F computed 1}
			- 2 \int _y ^\eta
			\diff t\	\frac{1}{\xinyp} 
			\lf(
			-\Om _0 t ^2
			-\tx\frac{1}{3}\Om _0 \eps ^2 t ^3
			+\betas \eps ^2 t
			\ri) \gcrm  \gcrm' =\frac{4\Om _0}{\alpha ^2}
			\int _y ^\eta
			\diff t \: (1 + \eps^2 t) \:
			U _\betas(t)
			\gcrm
			\gcrm'
			\\
			+  2\eps ^2\int _y ^\eta
			\diff t\
			\tx\frac{1}{\xintp}
			\lf(
			\frac{1}{6} \alpha^2 t ^3
			- \frac{1}{2} \Omega_0 (s-2) \betas t 
			+\Om _0 \eps ^2 \betas t ^2-\frac{1}{2}\eps ^2 \betas^2 
			\ri)
			\gcrm\gcrm' \\
			= \frac{4\Om _0}{\alpha ^2}
			\int _y ^\eta
			\diff t \: (1 + \eps^2 t) \:
			U _\betas(t)
			\gcrm
			\gcrm' + \OO(\eps^2 \eta^7) \gvbm^2(y),
		}
		where we have used the bound \eqref{eq: derivative gcrm} and the monotonicity of $ \gcrm $ for $ y \geq y_{\betas} $. The first term on the r.h.s. can be rewritten by means of \eqref{eq: gvbm variational}:
		\bmln{
			\frac{4\Om _0}{\alpha ^2}
			\int _y ^\eta
			\diff t \: (1 + \eps^2 t) \:
			U _\betas(t)
			\gcrm
			\gcrm'
			=\frac{2\Om _0}{\alpha ^2}
			\int _y ^\eta
			\diff t \: (1 + \eps^2 t) \:
			\lf[
			\tx\frac{1}{2}
			\lf(
				\gcrm'
			\ri) ^2
			-\frac{1}{2\pi} \gcrm ^4
			+\gcrchem \gcrm ^2
			\ri]' 
			\\
			+\frac{2\Om _0 \eps ^2}{\alpha ^2}
			\int _y ^\eta	
			\diff t\ 
			\lf(
			\gcrm '
			\ri) ^2
			-\frac{4\Om _0 \eps ^2}{\alpha ^2}
			\int _y ^\eta
			\diff t \: (1 + \eps^2 t) \:
			t ^3 v(t) \gcrm \gcrm'=
			\\
			=\frac{2\Om _0}{\alpha ^2}
			\lf(
			1
			+\eps ^2 \eta
			\ri)
			\lf[
			-\tx\frac{1}{2\pi}\gcrm ^4(\eta)
			+\gcrchem \gcrm ^2(\eta)
			\ri]
			-\frac{2\Om _0}{\alpha ^2}\xinyp
			\lf[
			\tx\frac{1}{2}
			\lf(
				\gcrm'(y)
			\ri) ^2
			-\frac{1}{2\pi}
			\gcrm ^4(y)
			+\gcrchem \gcrm ^2(y)
			\ri]-
			\\
			+\frac{2\Om _0 \eps ^2}{\alpha ^2}
			\int _y ^\eta
			\diff t\
			\lf[
			\tx\frac{1}{2}
			\lf(
				\gcrm' 
			\ri) ^2
			+\frac{1}{2\pi}\gcrm ^4 
			-\gcrchem \gcrm ^2 
			\ri]
			-\frac{4\Om _0 \eps ^2}{\alpha ^2}
			\int _y ^\eta
			\diff t \: (1 + \eps^2 t) \:
			t ^3 v(t) \gcrm  \gcrm'	\\
			= \frac{2\Om _0}{\alpha ^2}
			\gcrchem \gcrm ^2(\eta)
			-\frac{2\Om _0}{\alpha ^2}
			\lf[
			\tx\frac{1}{2}
			\lf(
				\gcrm'(y)
			\ri) ^2
			-\frac{1}{2\pi}
			\gcrm ^4(y)
			+\gcrchem \gcrm ^2(y)
			\ri] + \OO(\eps^2\eta^7) \gcrm^2(y).
		}
		Putting together the above estimate with \eqref{eq: F computed 0} and \eqref{eq: F computed 1}, we obtain the result.
	\end{proof}

\nt
Thanks to Proposition \ref{pro: alt expression F}, the cost function $ K$ can also be expressed as
\beq
	\label{eq: K alternative}
	K(y) = - \frac{\Om _0}{\alpha ^2} \lf(
				\gcrm' 
			\ri) ^2 
			+ \frac{2\Om _0}{\alpha ^2}\lf[ \frac{\alpha^2}{4 \Omega_0} + \frac{1}{2} \alpha^2 y^2 
			-\frac{1}{2\pi} \gcrm^4
			+\gcrchem 
			\ri] \gcrm ^2 
			+
			\begin{cases}
				R_{+}(y) + R_{+},	&	\mbox{if } y \geq y_\betas,	\\
				R_{-}(y) + R_{-},	&	\mbox{if } y \leq y_\betas.
			\end{cases}
\eeq
This alternative expression will play an important role in the proof of its positivity, exactly as for $ \kgv $. Another important ingredient of the proof is also the closeness of $ K $ to $ \kgv $ as $ \eps \to 0 $:
	
	\begin{lem}
		\label{lem: K close kgv}
		\mbox{}	\\
		For any $ \Omega_0 > 0 $ and $ y \in \annt $
		\beq
			\label{eq: K close kgv}
			K(y) - \kgv(y) = \OO(\eps^2|\log\eps|^{\infty}).
		\eeq
	\end{lem}
	
	\begin{proof}
		The result is a direct consequence of the pointwise estimate \eqref{eq: infty est diff gvm}.
	\end{proof}
	
\nt
The above Lemma in combination with Proposition \ref{pro: kgv positive} might seem to give also the positivity of $ K $ inside $ \annt $. However this is not the case because, although we proved that $ \kgv $ is positive on the whole real line, we did not provide any lower bound to it. In fact even by just looking at its minima, one could conclude from \eqref{eq: kgv crucial pro} that $ \kgv(y_2) \geq \kgv(y_1) \gvm^2(y_2)/\gvm^2(y_1) $, where $ y_{2}, y_1 $ are the positions of the minimum point and the preceding maximum point (consider for simplicity the half-line $ \R^+ $). Now even if $ \kgv(y_1) > C > 0 $ as it occurs for instance at the origin, the ratio between the densities can become extremely small in $ \ann $. In addition to that the inequality holds true only for the minima of $ \kgv $ and it might be that it has no minimum inside $ \annt $ or $ \ann $, in which case we only know that it is positive there, but without any meaningful lower bound.

In fact we will be able to prove positivity of $ K $ only in domain strictly smaller than $ \ann $, because of the additional constant terms $ R_{\pm} $ in \eqref{eq: K alternative}, irrespective of their smallness. We thus set
\beq
	\label{eq: annm}
	\annm : = 
	\lf\{
		y \in [-\eta,\eta]\: \big|\:
		\gcrm ^2 (y)
		\ge \eta ^6 \max\lf\{
			\gcrm ^2 (\eta),
			\gcrm ^2 (-\eta)
		\ri\}
	\ri\}.
\eeq
By monotonicity of $ \gcrm $ for large $ y $ is easy to see that $ \annm = [- y_{-}, y_{+}] $ with $ y_{\pm} \to \infty $ as $ \eps \to 0 $. Notice also that $ \gcrm $ is very small at the boundary of $ \annm $, although not as small as $ \gcrm(\eta) $.

We can now state the main result of this Section:

	\begin{pro}[Positivity of $ K $]
		\label{pro: K positive}
		\mbox{}	\\
		If $ \Omega_0 > \Omegac $ as $ \eps \to 0 $, 
		\beq
			K(y) > 0,	\qquad		\mbox{for any } y \in \annm.
		\eeq
	\end{pro}

	
%

	\begin{proof}
		As in the proof of positivity of $ \kgv $ in Proposition \ref{pro: kgv positive} the key idea is to show that positivity at the origin is indeed sufficient to get the result. This in turn is easily inherited from positivity of $ \kgv $ whenever $ \Omega_0 > \Omegac $, via \eqref{eq: K close kgv}: by \eqref{eq: kgv0}
		\bdm
			\kgv(0) > C > 0,
		\edm
		but
		\bdm
			K(0) - \kgv(0) = \OO(\eps^2|\log\eps|^{\infty}),
		\edm
		and thus $ K(0) > C > 0 $ for a possibly different constant $ C $.
		
		The rest of the proof follows the same line of reasoning of the proof of Proposition \ref{pro: kgv positive}. There are however two complications: first we have two alternative expressions of $ K $ in $ [y_{\betas},\eta] $ and $ [-\eta,y_{\betas}] $ respectively. Recall that $ y_{\betas} = o(1) $ denotes the unique maximum point of $ \gcrm $. Second the presence of the constant terms $ R_{\pm} $ in \eqref{eq: K alternative} is very annoying and in fact it is responsible of the restriction to $ \annm $. 
		
		In order to handle the first issue it is sufficient to take into account the two intervals $ [y_{\betas},\eta] $ and $ [-\eta,y_{\betas}] $ separately and use a different expressions for $ K $ (see \eqref{eq: K alternative}).
		
		The second issue on the other hand leads to the introduction of the modified cost function
		\beq
			\label{eq: tilde K}
			\kt(y) : = K(y)  - \deps \gcrm^2(y) - R_{\pm}
		\eeq
		for some 
		\beq
			\label{eq: deps}
			 0 < \deps \ll \eta^{-2} \ll 1
		\eeq
		to be chosen later. Here we have used a compact notation to mean that we subtract $ R_+ $ (resp. $ R_- $) in $ [y_{\betas},\eta] $ (resp. $ [-\eta,y_{\betas}] $).
		
		Now we observe that if $ \Omega_0 > \Omegac $
		\beq
			\kt(y_{\betas}) = \kt(0) + o(1) = K(0) + o(1) > 0,
		\eeq
		thanks to the pointwise estimate \eqref{eq: K close kgv} and since $ \lf\| \gcrm \ri\|_{\infty} \leq C $. It is interesting to remark that this is the only point in the proof where we use the condition $ \Omega_0 > \Omegac $, although several later estimates are affected by this one. 
		Moreover at the boundary of the domain we have
		\beq
			\kt(\pm\eta) = \lf( \tx\frac{1}{2} - \deps + \OO(\eps^3 \eta^7) \ri) \gcrm^2(\pm\eta) > 0.
		\eeq
		Therefore in order to exclude that $ \kt $ becomes negative, it suffices to prove that it is positive at any possible {\it global} minimum point $ - \eta < y_m < \eta $.
		
		We claim that, for $ \Omega_0 > \Omegac $, any global minimum point $ y_m $ of $ \kt $ must satisfy the condition
		\beq
			\label{eq: ym condition}
			|y_m| \gg 1.
		\eeq
		The reason is the pointwise estimate \eqref{eq: K close kgv} and the observation contained in Proposition \ref{pro: kgv positive}: $ \kt $ and $ \kgv $ are pointwise close and therefore we would have 
		\bdm
			\kgv(y_m) \leq \min [K(-\eta),K(\eta)] + o(1) \leq \half \min [\gcrm^2(-\eta),\gcrm^2(\eta)]+ o(1) = o(1),
		\edm
		 which in turn implies $ \kgv(y_m) \leq 0 $ since $ \kgv $ is independent of $ \eps $. For $ \Omega_0 > \Omegac $ this contradicts the statement of Proposition \ref{pro: kgv positive}.
		
		The key point in the proof is the following property: let $ y_m $ be a point where $ \kt $ reaches its global minimum $ \kt(y_m) < \kt(y_{\betas}) $ (otherwise there would be nothing to prove) and $ y_M $ any maximum point of $ \kt $ such that $ y_M < y_m $, if $ y_m > y_{\betas} $, or $ y_M > y_m $ in the opposite case $ y_m < y_{\betas} $. Notice that such a maximum needs not to be the global maximum but its existence is a consequence of smoothness of $ \kt $ and the inequalities
		\bdm
			\kt(\pm\eta) < \kt(y_{\betas}), \qquad	\kt(y_m) < \kt(y_{\betas}).  
		\edm
		Then we are going to prove that
		\beq
			\label{eq: K crucial pro 1}
			\frac{\kt(y_m)}{\gcrm^2(y_m)} \geq \frac{\kt(y_M)}{\gcrm^2(y_M)} + o(1).
		\eeq
		Now suppose that this is true, then we can pick a maximum point $ y_M $ of $ \kt $ where $ \kt(y_M) \geq \kt(y_{\betas}) > C > 0 $. In addition it must be 
		\beq
			y_M = \OO(1),
		\eeq
		because $ \kt(y) \leq C g^2(y) $ and the decay estimate \eqref{eq: better estimate on gvbm} or the pointwise estimate \eqref{eq: infty est diff gvm} imply that $ \kt(y) = o(1) $, if $ |y| \gg 1 $. Hence by the lower bound \eqref{eq: better estimate on gvbm} $ \gcrm(y_M) \geq C > 0 $, \eqref{eq: K crucial pro 1} yields
		\beq
			\label{eq: K crucial pro 2}
			\kt(y_m) \geq \gcrm^2(y_m) \lf( C^{-2} \kt(y_{\betas}) + o(1) \ri) \geq C_0 \gcrm^2(y_m) > 0
		\eeq
		for some $ C_0 > 0 $. In fact we have obtained something more: for any $ y \in \ann $ either $ \kt(y) \geq \min \{ \kt(-\eta), \kt(\eta) \} > 0  $ or
		\bdm
			\frac{\kt(y)}{\gcrm^2(y)} \geq \frac{\kt(y_m)}{\gcrm^2(y)} \geq C \frac{\gcrm^2(y_m)}{\gcrm^2(y)} > 0.
		\edm
		Either way $ \kt $ is positive everywhere in $ [-\eta,\eta] $. Moreover the positivity of $ \kt $ implies that
		\bdm
			K(y) > R_{\pm} + \deps \gcrm^2(y) \geq 0,	\qquad \mbox{if } \gcrm^2(y) \geq \deps^{-1} |R_+|,
		\edm
		for any $ y \in \ann $. If now we restrict the inequality to $ \annm $ and we choose, e.g., $ \deps = \eta^{-3} $, the estimates \eqref{eq: tails} imply that inside $ \annm $
		\bdm
			\gcrm^2(y) \geq \eta^6 \max\lf\{\gcrm^2(-\eta), \gcrm^2(\eta) \ri\} \gg C\eta^{5} \max\lf\{\gcrm^2(-\eta), \gcrm^2(\eta) \ri\} \geq \deps^{-1} |R_{\pm}|,
		\edm
		so that $ K(y) $ is strictly positive for any $ y \in \annm $. In fact a closer look to the chain of inequalities reveals that we have proven something more, i.e., for $ \eps $ small enough
		\beq
			\label{eq: lb K}
			K(y) \geq \eta^{-3} \gcrm^2(y),	\qquad	\mbox{for any } y \in \annm.
		\eeq
		
		We have now to prove \eqref{eq: K crucial pro 1}. Recall the assumption: we have a global minimum of $ \kt $ at $ y_m $ and a maximum at $ y_M $, which is on the left (resp. right) of $ y_m $, if $ y_m > y_{\betas} $ (resp. $ y_m < y_{\betas} $). The idea is the same used in the proof of Proposition \ref{pro: kgv positive}: the derivative of $ \kt $, i.e.,
		\beq
			\label{eq: kt derivative}
			\kt^{\prime}(y) =\lf[
			(1 - \deps) \gcrm'
			+ \tx\frac{1}{\xiny}
			\lf(
				2\Om _0 y
				+\Om _0 \eps ^2 y ^2
				- \betas \eps ^2
			\ri)\gcrm
		\ri] \gcrm,
		\eeq
		must vanish both at $ y_m $ and $ y_M $ and therefore
		\beq
			\label{eq: extreme K}
			\gcrm^{\prime}(y_{m,M}) = - \frac{1}{(1 - \deps)(1 + \eps^2 y_{m,M})}\lf(
				2\Om _0 y_{m,M}
				+\Om _0 \eps ^2 y_{m,M} ^2
				- \betas \eps ^2
			\ri)\gcrm(y_{m,M}).
		\eeq
		The second derivative of $K $ can be computed as well:
		\beq
			\kt''(y) =
			(1 - \deps) \lf(\gcrm'\ri)^2 + 4 \Omega_0 y \gcrm \gcrm' + \lf[ (1 - \deps)\lf( \alpha^2 y^2 + \tx\frac{2}{\pi} \gcrm^2 - 2 \gcrchem \ri) + 2 \Omega_0 + \OO(\eps^2 \eta^5) \ri] \gcrm^2, 		
		\eeq
		so that at any extreme point of $ \kt $, one has
		\beq
			\label{eq: kt primeprime ext}
			\kt''(y_{m,M}) = \lf[ \Omega_0^2 (s-2) y_{m,M}^2 + 2 \Omega_0 + \tx\frac{2}{\pi} \gcrm^2(y_{m,M}) - 2 \gcrchem + o(1) \ri] \gcrm^2(y_{m,M}),
		\eeq
		where we have exploited the condition \eqref{eq: deps}. Similarly by \eqref{eq: K alternative} we get
		\beq
			\label{eq: kt ext}
			\kt(y_{m,M}) = \lf[ \Omega_0 \tx\frac{s-2}{s+2} y_{m,M}^2  + \tx\frac{1}{2} - \tx\frac{1}{\pi \Omega_0(s+2)} \gcrm^2(y_{m,M}) - \frac{2}{\Omega_0(s+2)} \gcrchem + o(1) \ri] \gcrm^2(y_{m,M}),
		\eeq
		and a direct comparison between \eqref{eq: kt primeprime ext} and \eqref{eq: kt ext} yields
		\beq
			\label{eq: key identity}
			\frac{\kt(y_{m,M})}{\gcrm^2(y_{m,M})} = \frac{s-2}{s+2} -\frac{\gcrm^2(y_{m,M})}{\pi \Omega_0 (s+2)} + \frac{\kt''(y_{m,M}) }{\Omega_0 (s+2) \gcrm^2(y_{m,M})} + o(1).
		\eeq
		Now this is the key identity because by assumption (recall also \eqref{eq: ym condition})
		\bdm
			\kt''(y_{M}) \leq 0 \leq \kt''(y_{m}),	\qquad		\gcrm(y_m) < \gcrm(y_M),
		\edm
		so that
		\bdm
				\frac{\kt(y_{m})}{\gcrm^2(y_{m})} \geq \frac{\kt(y_{M})}{\gcrm^2(y_{M})} + o(1),
		\edm
		i.e., \eqref{eq: K crucial pro 1} is proven. Note that the fact that we have two different explicit expressions of $ \kt $ for $ y > y_{\betas} $ and $ y < y_{\betas} $ did not affect the proof, because the difference between the two expressions is $ o(1) $ and therefore can be included in the error term.
	\end{proof}

\section{Energy Asymptotics}
\label{sec: energy}

We attack in this Section the proof of Theorem \ref{teo: energy asymptotics}, which will imply the main result of the paper. The result is obtained by combining upper (Proposition \ref{pro: upper bound}) and lower (Proposition \ref{pro: lower bound}) bounds on $ \gpe $.

\subsection{Upper Bound}
\label{sec: energy ub}

The upper bound on $ \gpe $ is stated in next

	\begin{pro}[GP energy upper bound]
		\label{pro: upper bound}
		\mbox{}	\\
		As $ \eps \to 0 $,
		\beq
			\label{eq: upper bound}
			\gpe \leq \frac{\gcre}{\eps^4} + \OO(1).
		\eeq
	\end{pro}
	
	\begin{proof}
		The proof is rather simple because it is sufficient to test $ \gpf $ on suitable trial function of the form
		\beq
			\psitrial(\xv) :
			=\frac{1}{\sqrt{2\pi}\eps}
			\gtrial\lf(
			\tx\frac{|\xv|-1}{\eps ^2}
			\ri)
			\exp \lf\{i\lf\lfloor\Om +\betas\ri\rfloor\theta \ri\},
		\eeq
		where $ \gtrial $ coincides up to a normalization constant with $ \gcrm $ within $ \ann $ and is suitably regularized outside. The calculation is rather straightforward and we omit it for the sake of brevity. Note that the remainder $ \OO(1) $ is entirely due to the fact that the phase $ \Omega + \betas $ might not be an integer number. Otherwise one would obtain a much better error term $ \OO(\eps^4) $.
	\end{proof}

\subsection{Lower Bound}
\label{sec: energy lb}

A lower bound for $ \gpe $ matching the upper bound of Proposition \ref{pro: upper bound} is formulated in next

	\begin{pro}[GP energy lower bound]
		\label{pro: lower bound}
		\mbox{}	\\
		If $ \Omega_0 > \Omegac $, as $ \eps \to 0 $,
		\beq
			\label{eq: lower bound}
			\gpe \geq \frac{\gcre}{\eps^4} + \OO(\eps^{\infty}).
		\eeq
	\end{pro}
	
	\begin{proof}
		We first restrict the integration in the GP energy functional to the domain $ \ann $ (recall its definition in \eqref{eq: ring}: to this purpose we just have to observe that all the three terms in the GP energy functional are pointwise positive and thus we can simply drop their integrals outside $ \ann $. Of course $ \gpm $ is not normalized in $ L^2(\ann) $ but the exponential decay proven in Proposition \ref{pro: gpm exp small} guarantees that
		\beq
			\lf\| \gpm \ri\|_{L^2(\ann)} = 1 + \OO(\eps^{\infty}),
		\eeq
		by taking $ \eta_0 $ large enough.
		
		The first step in the proof is a splitting of the energy, in order to extract the leading order term $ \gcre/\eps^4 $. This is now rather standard and we do not spell all the details of the computation. We just note that one sets
		\beq
			\label{def: u}
			\gpm(\xv)
			=: \frac{1}{\sqrt{2\pi}\eps}
			u\lf(
			x,\vartheta
			\ri)
				\gcrm \lf(
			\tx\frac{x-1}{\eps ^2}
			\ri)
		e^{i\lf(\Om +\betas\ri)\theta}.
		\eeq
		Since $ \Omega + \betas $ needs not to be an integer, $ u $ is not single-valued in general, but
		\beq
			\label{eq: u semiperiodic}	
			u(x,\vartheta+2 k \pi) = e^{-i 2\pi k \lf(\Om +\betas\ri)} u(x,\vartheta),
		\eeq
		for any $ k \in \Z $. A part from that $ u $ is finite for any $ \xv \in \ann $ thanks to the strict positivity of $ \gcrm $. A long but simple computation using the variational equation for $ \gcrm $ gives
		\beq
			\label{eq: splitting}
			\gpe \geq \frac{\gcre}{\eps^4} + \frac{\E[u]}{2 \pi \eps^2} + \OO(\eps^{\infty}),
		\eeq
		where the inequality is mainly due to the restriction of the integration domain and setting $ y = 1 + \eps^2 x $ for short
		\beq
			\label{eq: eu}
			\E[u] = \int _\ann
			\diff \xv\
			\gcrm ^2(y)
			\lf\{
			\tx\frac{1}{2}
			\lf|
				\nabla u
			\ri| ^2
			+{\bf a}
			\cdot {\bf j} _u 
			+\frac{1}{2\pi \eps ^4} \gcrm ^2(y)
			(
			1
			-|
				u
			| ^2
			)^2
			\ri\},
		\eeq
		\beq
			{\bf a}(x):=\lf(
			\frac{\Om + \betas}{x}
			-\Om x
		\ri){\bf e_\vartheta},
		\eeq	
		and the superfluid current is defined in \eqref{eq: current}. The rest of the proof is devoted to prove that
		\beq
			\E[u] \geq \OO(\eps^{\infty}).
		\eeq
		
		In order to exploit the cost function trick mentioned in Section \ref{sec: heuristics} and the positivity of $ K $ proven in Proposition \ref{pro: K positive}, we need to restrict again the integration domain in $ \E[u] $ to $ \annmd \subset \ann $, where
		\bdm
			\annmd = \lf\{ \xv \in \R^2 \: \big| \: \lf|1 - |\xv| \ri|/\eps^2 \in \annm \ri\}.
		\edm 
		The only annoying term is the only one which is not positive, i.e., the second one in \eqref{eq: eu}:
		\bml{
			\label{eq: ang momentum layer}
			\lf| \int_{\ann \setminus \annmd} \diff \xv \: \gcrm^2(y) \: {\bf a} \cdot {\bf j} _u \ri| \leq \lf\| {\bf a} \ri\|_{L^{\infty}(\ann)} \int_{\ann \setminus \annmd} \diff \xv \: \gcrm^2(y) \lf|u\ri| \lf| \nabla_{\vartheta} u \ri| \\
			\leq C \eta^2 \int_{\ann \setminus \annmd} \diff \xv \: \lf| \gpm \ri| \lf| \nabla_{\vartheta} \gpm \ri| 	
			\leq C \eps^2 \eta^4 \lf\| \gpm \ri\|_{L^{\infty}(\ann\setminus \annmd)} \lf\| \nabla \gpm \ri\|_{L^{\infty}(\ann)}	\\
			 \leq C \eps^{-4} \eta^4 \lf\| \gpm \ri\|_{L^{\infty}(\ann\setminus \annmd)},
			}
			where we have used the bound $ \lf\| \nabla \gpm \ri\|_{L^{\infty}(\ann)} \leq C \eps^{-6} $, following from
			\beq
				\label{eq: GN}
				\lf\| \nabla \psi \ri\|_{\infty} \leq C \lf( \lf\| \Delta \psi \ri\|_{\infty}^{1/2} \lf\| \psi \ri\|_{\infty}^{1/2} + \lf\| \psi \ri\|_{\infty} \ri),
			\eeq
			which can be proven from Gagliardo-Nirenberg inequalities exactly as in \cite[Lemma 5.1]{CRY}. However the lower bound \eqref{eq: better estimate on gvbm} easily implies that if we set $ \annm =: [-y_-,y_+] $, then
			\bdm
				y_{\pm} = \eta(1 + o(1)),
			\edm
			so that 
			\beq
				\lf| \gpm \ri|_{\partial \annmd} \leq \OO(\eps^{\infty}),
			\eeq
			again by \eqref{eq: gpm exp small} and the arbitrariness in the choice of $ \eta_0 $. Hence \eqref{eq: ang momentum layer} yields an error which can be made smaller than any power of $ \eps $ and we get the lower bound
			\beq
				\E[u] \geq  \int _{\annmd}
				\diff \xv\
				\gcrm ^2(y)
				\lf\{
				\tx\frac{1}{2}
				\lf|
				\nabla u
				\ri| ^2
				+{\bf a}
				\cdot {\bf j} _u 
				+\frac{1}{2\pi \eps ^4} \gcrm ^2(y)
				(
				1
				-|
				u
				| ^2
				)^2
				\ri\} + \OO(\eps^{\infty}).
			\eeq
			
			We can now finally integrate by the angular momentum term by using the potential function $ F $ defined in \eqref{eq: potential f}: it is trivial to verify that
			\beq
				2 \gcrm^2\lf(\tx\frac{x - 1}{\eps^2}\ri) {\bf a}(x) = - \partial_x F \lf( \tx\frac{x - 1}{\eps^2} \ri) {\bf e}_{\vartheta},
			\eeq
			so that
			\bml{
				\label{eq: int by parts}
				 \int_{\annmd} \diff \xv \: \gcrm^2(y) \: {\bf a} \cdot {\bf j} _u = - \frac{1}{2} \int_0^{2\pi} \diff \vartheta \int_{1 - \eps^2 y_-}^{1 + \eps^2 y_+} \diff x \: \partial_x F \lf( \tx\frac{x - 1}{\eps^2} \ri)  \Re \lf[ i u(x,\vartheta) \partial_{\vartheta} u^*(x,\vartheta) \ri] 	\\
				 = \frac{1}{2} \int_0^{2\pi} \diff \vartheta \int_{1 - \eps^2 y_-}^{1 + \eps^2 y_+} \diff x \: F\lf( \tx\frac{x - 1}{\eps^2} \ri) \Re \lf[ i \partial_x u(x,\vartheta) \partial_{\vartheta} u^*(x,\vartheta) + i u(x,\vartheta) \partial^2_{x,\vartheta} u^*(x,\vartheta) \ri]	\\
				  - \frac{1}{2} \int_0^{2\pi} \diff \vartheta \: \Big| F\lf( \tx\frac{x - 1}{\eps^2} \ri) \Re \lf[ i u(x,\vartheta) \partial_{\vartheta} u^*(x,\vartheta) \ri] \Big|_{1-\eps^2 y_-}^{1 + \eps^2 y_+}.
			}
			The boundary term can be easily proven to provide an exponentially small correction: consider, e.g., the term at $ 1 + \eps^2 y_+ $, since $ \lf| F(y_{\pm}) \ri| \leq C \eta^8 \gcrm^2(\pm\eta) $, one can reconstruct a term, which can be bounded exactly as \eqref{eq: ang momentum layer}. The result is an error $ \OO(\eps^{\infty}) $. The rest is integrated by parts once more but this time w.r.t. $ \vartheta $:
			\bmln{
				\frac{1}{2} \int_0^{2\pi} \diff \vartheta \int_{1 - \eps^2 y_-}^{1 + \eps^2 y_+} \diff x \: F\lf( \tx\frac{x - 1}{\eps^2} \ri) \Re \lf[ i u(x,\vartheta) \partial^2_{x,\vartheta} u^*(x,\vartheta) \ri]		\\
				 = - \frac{1}{2} \int_0^{2\pi} \diff \vartheta \int_{1 - \eps^2 y_-}^{1 + \eps^2 y_+} \diff x \: F\lf( \tx\frac{x - 1}{\eps^2} \ri) \Re \lf[ i \partial_{\vartheta} u(x,\vartheta) \partial{x} u^*(x,\vartheta) \ri],
			}
			where the vanishing of boundary terms is due to the periodicity of $ u^* \partial_x u $ and  its complex conjugate (compare with \eqref{eq: u semiperiodic}). Altogether we have thus obtained that
			\bml{
				 \int_{\annmd} \diff \xv \: \gcrm^2(y) \: {\bf a} \cdot {\bf j} _u = \int_{\annmd} \diff \xv \: F\lf( \tx\frac{x - 1}{\eps^2} \ri) \Re \lf[ i \nabla_x u(x,\vartheta) \nabla_{\vartheta} u^*(x,\vartheta) \ri]	 + \OO(\eps^{\infty})\\
				  \geq - \int_{\annmd} \diff \xv \: \lf| F\lf( \tx\frac{x - 1}{\eps^2} \ri) \ri| \lf| \nabla u \ri|^2 + \OO(\eps^{\infty}) = - \int_{\annmd} \diff \xv \: F\lf( \tx\frac{x - 1}{\eps^2} \ri) \lf| \nabla u \ri|^2  + \OO(\eps^{\infty}),
			}
			and therefore
			\bml{
				\label{eq: lb final}
				\E[u] \geq \int_{\annmd} \diff \xv \: \lf\{ K\lf( \tx\frac{x - 1}{\eps^2} \ri) \lf| \nabla u \ri|^2 + \frac{1}{2\pi \eps ^4} \gcrm ^2(y)
				(
				1
				-|
				u
				| ^2
				)^2
				\ri\} + \OO(\eps^{\infty})	\\
				\geq \eta^{-3} \int_{\annmd} \diff \xv \: \gcrm^2(y) \lf| \nabla u \ri|^2 + \OO(\eps^{\infty}) \geq \OO(\eps^{\infty}),
			}
			thanks to Proposition \ref{pro: K positive} and in particular \eqref{eq: lb K}.			
	\end{proof}

\section{Giant Vortex Transition}
\label{sec: gv transition}

In this Section we prove the results regarding absence of vortices and total vorticity of the condensate.

	\begin{proof}[Proof of Theorem \ref{teo: no vortices}]	
		Combining \eqref{eq: lb final} with \eqref{eq: splitting} and the upper bound proven in Proposition \ref{pro: upper bound}, we get
		\beq
			\int_{\annmd} \diff \xv \: \gcrm ^4(y)
				(
				1
				-|
				u
				| ^2
				)^2 = \OO(\eps^6),
		\eeq
		which already means that $ |u| $ can not differ too much from $ 1 $. To deduce the pointwise estimate of Theorem \ref{teo: no vortices}, we need to combine this with an estimate of $ \lf\| \nabla  u \ri\|_{\infty} $.
		
		As in \cite[Lemma 4.3]{CRY} we obtain from \eqref{eq: GP variational} and \eqref{eq: gvbm variational} the following variational equation for $ u $:
		\beq
			\label{eq: u variational}
			-\half \gvbm \Delta u
			-\frac{1}{\eps ^2} \gcrm' \partial _x u
			-i \gcrm {\bf a} \cdot \nabla u
			+\frac{1}{\pi \eps ^4} \gcrm ^3
			\lf(
			\lf|
				u
			\ri| ^2
			-1
			\ri) u
			=\lf(
			\gpchem
			-\frac{1}{\eps ^4} \gcrchem
			\ri) \gcrm u,
		\eeq
		which yields (recall the definition of $ \ab \subset \annmd $ in \eqref{eq: bulk})
		\bdm
			\lf\| \Delta u \ri\|_{L^{\infty}(\ab)} \leq  C\lf[
			\eps ^{-2} \eta
			\lf\|
				\nabla u
			\ri\| _{L ^\infty(\ab)}
			+\eps ^{-4} \eta^{3a}
			\ri].
		\edm
		Now using the elliptic estimate \eqref{eq: GN} we conclude that
		\beq
			\label{eq: estimate on grad u}
			\lf\|
				\nabla u
			\ri\| _{L ^\infty(\ab)}
			= \OO\lf(\eps ^{-2} \eta ^{1+\frac{3a}{2}}\ri).
		\eeq
	
		Suppose now that it exists $ \xv _0 \in \ab $ such that $	\lf|
			u(\xv _0)
			-1
		\ri|
		\geq \eps^{1/2} |\log\eps|^b
		$,
		for some $ b > 0  $ to be chosen later. Then from \eqref{eq: estimate on grad u} we get that
		\[
			\tx
			\lf|
			|u|
			-1
			\ri|
			\ge \half \eps ^{1/2} |\log\eps|^b,	
			\qquad \mbox{for }
			\xv \in \ba_{\varrho}\lf(
			\xv _0\ri)	\cap \anna,
		\]
		with $ \varrho = \eps^{5/2} |\log\eps|^{b - 1 - \frac{3a}{2}} $, and
		\bdm
			\OO(\eps^6) =  \int _{\ab \cap  \ba_{\varrho}\lf(
			\xv _0\ri)}
			\diff \xv\
			\gcrm ^4(y) \lf(
			1 - |u| ^2
			\ri) ^2
			\ge C \eps^{6} |\log\eps|^{2b - 7a - 2},
		\edm
		which is a contradiction for all $ b \geq 4a - 1 $.
	\end{proof}

\nt
We now focus on the proof of Theorem \ref{teo: winding number} and for later purposes we state a useful Lemma, which is the analogue of \cite[Lemma 3.5]{CPRY3}:

	\begin{lem}
		\label{lem: deg u estimate}
		\mbox{}	\\
		Let $ \Omega_0 > \Omegac $ and $ R $ be a radius satisfying $ R = 1 +\bigO{\eps ^2} $, then 
		\beq
			\label{eq: winding u}
			\lf|
			\deg\lf(
				u,
				\partial \ba _R
			\ri)
			\ri|	= \OO(1)
		\eeq
	\end{lem}

	\begin{proof}
		We use a smooth radial cut-off function $ \chi $ with support in $ [\widetilde{R},R] $ such that $ \chi(\widetilde{R}) =0 $ and $ \chi(R) = 1 $, for some radius
	\bdm
		\widetilde{R}
		=R
		- c \eps ^2
	\edm
	with $ c>0 $. We also require that $ \lf|
			\chi
		\ri|
		\le 1 $ and $ \lf|
			\nabla \chi
		\ri| = \OO(\eps ^{-2}) $. Then by Stokes formula
		\bml{
			\deg\lf(
			u,
			\partial \ba _R
			\ri)
			=\frac{1}{\pi}
			\int _{\partial \ba _R}
			\diff \sigma \: \Im
			 \lf(
			\frac{\nabla_{\vartheta} u}{u}
			\ri) = \frac{1}{\pi}
			\int _{\partial \ba _R}
			\diff \sigma \: \chi(R) \: \Im
			 \lf(
			\frac{\nabla_{\vartheta} u}{u}
			\ri) \\
			=
			\frac{1}{\pi} \int _{\ba_R\setminus \ba_{\widetilde{R}} }
			\diff \xv \:	\nabla ^\perp \chi \cdot 
			\Im \lf(
			\frac{\nabla u}{u}
			\ri).
		}
		Therefore
		\beq
			\label{eq: deg est}
			\lf|
			\deg\lf(
				u,
				\partial \ba _R
			\ri)
			\ri|
			\le \frac{C}{\eps ^2}
			\int _{\ba_R\setminus \ba_{\widetilde{R}} }
			\diff \xv\
			\frac{\lf|\nabla u\ri|}{|u|}
			\leq \frac{C}{\eps^2} \lf| \ba_R\setminus \ba_{\widetilde{R}} \ri|^{1/2} \lf\| \nabla u \ri\|_{L^2(\ba_R\setminus \ba_{\widetilde{R}})},
		\eeq
		where we used that $ \lf\| 1 - |u| \ri\|_{L^{\infty}(\ba_R\setminus \ba_{\widetilde{R}})} = o(1) $. Now the result proven in Proposition \ref{pro: kgv positive} in fact says that for any $ \xv \in \ba_R\setminus \ba_{\widetilde{R}} $ and for $ \Omega_0 > \Omegac $, there exists a constant $ C > 0  $ such that
		\bdm
			\kgv\lf(\tx\frac{x-1}{\eps^2}\ri) \geq C,
		\edm
		and thanks to \eqref{eq: K close kgv} the same inequality holds true for $ K $, i.e.,
		\beq
			K\lf(\tx\frac{x-1}{\eps^2}\ri) \geq C > 0, 	\qquad		\mbox{for any } \xv \in \ba_R\setminus \ba_{\widetilde{R}}.
		\eeq
		Going back to \eqref{eq: lb final} this yields
		\bdm
			\OO(\eps^2) \geq \int_{\annmd} \diff \xv \:  K\lf( \tx\frac{x - 1}{\eps^2} \ri) \lf| \nabla u \ri|^2 \geq  \int_{\ba_R\setminus \ba_{\widetilde{R}}} \diff \xv \:  K\lf( \tx\frac{x - 1}{\eps^2} \ri) \lf| \nabla u \ri|^2 \geq C \lf\| \nabla u \ri\|_{L^2(\ba_R\setminus \ba_{\widetilde{R}})},
		\edm
		which gives the result once plugged into \eqref{eq: deg est}.
	\end{proof}

	\nt
	We are now in position to complete the estimate of the winding number of $ \gpm $:
	
	\begin{proof}[Proof of Theorem \ref{teo: winding number}]	
		We follow \cite[proof of Theorem 1.5]{CPRY3}. The positivity of $ |\gpm| $ on $ \partial \ba _R $ is guaranteed for any radius $ R = 1 +\bigO{\eps ^2} $ thanks to \eqref{eq: no vortices}. A simple computation shows that
		\[
			\deg \lf(
			\gpm,
			\partial \ba_R
			\ri)
			=\Om 
			+\betas
			+\deg \lf(
			u,
			\partial \ba_R
			\ri),
		\]
		which yields the result in combination with \eqref{eq: winding u}.
	\end{proof}


\begin{thebibliography}{aaa99}

\bibitem[ARVK]{ARVK}	\textsc{J.R. Abo-Shaeer, C. Raman, J.M. Vogels, W. Ketterle}, Observation of Vortex Lattices in Bose-Einstein Condensates,  {\it Science} {\bf 292}, 476--479 (2001).










\bibitem[AS]{AS}	\textsc{M. Abramovitz, I.A. Stegun}, \emph{Handbook of Mathematical Functions: with Formulas, Graphs, and Mathematical Tables}, Dover, New York, 1965.

\bibitem[ABD]{ABD}  \textsc{A. Aftalion, X. Blanc, J. Dalibard}, Vortex Patterns in a Fast Rotating Bose-Einstein Condensate, \textit{ Phys. Rev. A} \textbf{71}, 023611 (2005).

\bibitem[ABN]{ABN}  \textsc{A. Aftalion, X. Blanc, F. Nier}, Vortex distributions in the lowest Landau level, \textit{ Phys. Rev. A} \textbf{73}, 011601 (2006).

\bibitem[AJR]{AJR} \textsc{A. Aftalion, R.L. Jerrard, J. Royo-Letelier}, Non-Existence of Vortices in the Small Density Region of a Condensate, \textit{J. Funct. Anal.} \textbf{260}, 2387--2406 (2011).


\bibitem[AD1]{AD1}	\textsc{X. Antoine, R. Duboscq}, GPELab, a Matlab Toolbox to Solve Gross-Pitaevskii Equations I: Computation of Stationary Solutions, {\it Comput. Phys. Commun.}, {\bf 185} (2014), 2969--2991.

\bibitem[AD2]{AD2} \textsc{X. Antoine, R. Duboscq}, GPELab, A Matlab Toolbox to Solve Gross-Pitaevskii Equations II: Dynamics and Stochastic Simulations, {\it Comput. Phys. Commun.}  {\bf 193} (2015), 95--117.









\bibitem[BSSD]{BSSD}\textsc{V. Bretin, S. Stock, Y. Seurin, J. Dalibard}, Fast Rotation of a Bose-Einstein Condensate, {\it Phys. Rev. Lett.} {\bf 92}, 050403 (2004).



\bibitem[BCPY]{BCPY} 	\textsc{J.-B. Bru, M. Correggi, P. Pickl, J. Yngvason}, The TF Limit for Rapidly Rotating Bose Gases in Anharmonic Traps, \emph{Commun. Math. Phys.} \textbf{280} (2008), 517--544.

\bibitem[CD]{CD} \textsc{Y. Castin, R. Dum}, Bose-Einstein Condensates with Vortices in Rotating Traps, \textit{Eur. Phys. J. D} \textbf{7} (1999), 399--412.

\bibitem[CHES]{CHES} \textsc{I. Coddington, P.C. Haljan, P. Engels, V. Schweikhard, S. Tung, E.A. Cornell},Experimental studies of equilibrium vortex properties in a Bose-condensed gas, \textit{Phys. Rev. A} \textbf{70}, 063607 (2004).

\bibitem[Co]{Co} 		\textsc{N.R. Cooper}, Rapidly Rotating Atomic Gases, \emph{Adv. Phys.} \textbf{57} (2008), 539--616. 

\bibitem[CDY1]{CDY1}   	\textsc{M. Correggi, T. Rindler-Daller, J. Yngvason}, Rapidly Rotating Bose-Einstein Condensates in Strongly Anharmonic Traps, \emph{J. Math. Phys.} \textbf{48} (2007), 042104.

\bibitem[CDY2]{CDY2}   	\textsc{M. Correggi, T. Rindler-Daller, J. Yngvason}, Rapidly Rotating Bose-Einstein Condensates in Homogeneous Traps, \emph{J. Math. Phys.} \textbf{48} (2007), 102103.

\bibitem[CPRY1]{CPRY1} \textsc{M. Correggi, F. Pinsker, N. Rougerie, J. Yngvason}, Critical Rotational Speeds in the Gross-Pitaevskii Theory on a Disc with Dirichlet Boundary Conditions, {\it J. Stat. Phys.} {\bf 143}, 261--305 (2011).

\bibitem[CPRY2]{CPRY2} \textsc{M. Correggi, F. Pinsker, N. Rougerie, J. Yngvason}, Rotating Superfluids in Anharmonic Traps: From Vortex Lattices to Giant Vortices, \emph{Phys.\ Rev.\ A} \textbf{84}, 053614 (2011). 

\bibitem[CPRY3]{CPRY3} \textsc{M. Correggi, F. Pinsker, N. Rougerie, J. Yngvason}, Critical Rotational Speeds for Superfluids in  Homogeneous Traps, {\it J. Math. Phys.} \textbf{53}, 095203 (2012).

\bibitem[CPRY4]{CPRY4} \textsc{M. Correggi, F. Pinsker, N. Rougerie, J. Yngvason}, Vortex Phases of Rotating Superfluids, {\it J. Phys.: Conf. Ser.} {\bf 414} (2013), 012034.

\bibitem[CPRY5]{CPRY5} \textsc{M. Correggi, F. Pinsker, N. Rougerie, J. Yngvason}, Giant Vortex Phase Transition in Rapidly Rotating Trapped Bose-Einstein Condensates, {\it Eur. Phys. J. Special Topics} {\bf 217} (2013), 183--188.

\bibitem[CR1]{CR1}		\textsc{M. Correggi, N. Rougerie},  Inhomogeneous Vortex Patterns in Rotating Bose-Einstein Condensates, {\it Commun. Math. Phys.} {\bf 321} (2013), 817--860.

\bibitem[CR2]{CR2}		\textsc{M. Correggi, N. Rougerie}, On the Ginzburg-Landau Functional in the Surface Superconductivity Regime, {\it Commun. Math. Phys.} {\bf 332} (2014), 1297--1343; erratum {\it Commun. Math. Phys.} {\bf 338} (2015), 1451--1452.

\bibitem[CR3]{CR3}		\textsc{M. Correggi, N. Rougerie}, Boundary Behavior of the Ginzburg-Landau Order Parameter in the Surface Superconductivity Regime, {\it Arch. Rational Mech. Anal.} {\bf 219} (2016), 553--606.

\bibitem[CRY]{CRY}	\textsc{M. Correggi, N. Rougerie, J. Yngvason}, The Transition to a Giant Vortex Phase in a Fast Rotating Bose-Einstein Condensate, {\it Commun. Math. Phys.} {\bf 303}, 451--508 (2011).

\bibitem[CY]{CY}		\textsc{M. Correggi, J. Yngvason}, Energy and Vorticity in Fast Rotating Bose-Einstein Condensates, {\it J. Phys. A: Math. Theor.} {\bf 41}, 445002 (2008).

\bibitem[Dan]{Dan} \textsc{I. Danaila}, Three-dimensional Vortex Structure of a Fast Rotating Bose-Einstein Condensate with Harmonic-plus-quartic Confinement, \textit{ Phys. Rev. A} \textbf{72}, 013605 (2005).

\bibitem[E]{E}	\textsc{L.C. Evans}, {\it Partial Differential Equation}, Graduate Studies in Mathematics {\bf 19}, AMS, Providence, 1998.



\bibitem[Fe1]{Fe1} 	\textsc{A.L. Fetter}, Rotating Trapped Bose-Einstein Condensates, \emph{Rev. Mod. Phys.} \textbf{81} (2009), 647--691.

\bibitem[Fe2]{Fe2}	\textsc{A.L. Fetter},  Rotating Vortex Lattice in a Bose-Einstein Condensate Trapped in Combined Quadratic and Quartic Radial Potentials, \emph{Phy. Rev. A} \textbf{64} (2001), 063608. 

\bibitem[FJS]{FJS} 	\textsc{A.L. Fetter, N. Jackson, S. Stringari},  Rapid Rotation of a Bose-Einstein Condensate in a Harmonic Plus Quartic Trap, \emph{Phys. Rev. A} \textbf{71} (2005), 013605. 

\bibitem[FB]{FB} 		\textsc{U.R. Fischer, G. Baym}, Vortex States of Rapidly Rotating Dilute Bose-Einstein Condensates, \emph{Phys. Rev. Lett.} \textbf{90} (2003), 140402.



\bibitem[IM1]{IM1} 	\textsc{R. Ignat, V. Millot}, The Critical Velocity for Vortex Existence in a Two-dimensional Rotating Bose-Einstein Condensate, \emph{J. Funct. Anal.} \textbf{233}, 260--306 (2006).

\bibitem[IM2]{IM2} 	\textsc{R. Ignat, V. Millot}, Energy Expansion and Vortex Location for a Two Dimensional Rotating Bose-Einstein Condensate, \emph{Rev. Math. Phys.} \textbf{18}, 119--162 (2006).

\bibitem[KTU]{KTU} 	\textsc{K. Kasamatsu, M. Tsubota, M. Ueda},  Giant Hole and Circular Superflow in a Fast Rotating Bose-Einstein Condensate, \emph{Phys. Rev. A} \textbf{66} (2002), 050606.

\bibitem[Ka]{Ka}	\textsc{A. Kachmar}, Energy of a rotating Bose-Einstein condensate in a harmonic trap, preprint {\it  arXiv:1306.3296 [math.AP]}.

\bibitem[KB]{KB} 		\textsc{G.M. Kavoulakis, G. Baym}, Rapidly Rotating Bose-Einstein Condensates in Anharmonic Potentials, \emph{New J. Phys.} \textbf{5} (2003), 51.1-51.11.

\bibitem[KF]{KF} 		\textsc{J.K. Kim, A.L. Fetter},  Dynamics of a Rapidly Rotating Bose-Einstein Condensate in a Harmonic Plus Quartic trap, \emph{Phys. Rev. A} \textbf{72} (2005), 023619. 





\bibitem[LM]{LM} 		\textsc{L. Lassoued, P. Mironescu}, Ginzburg-Landau Type Energy with Discontinuous Constraint, \emph{J. Anal. Math.} \textbf{77} (1999), 1--26.

\bibitem[LS]{LS} \textsc{E.H. Lieb, R. Seiringer}, Derivation of the Gross-Pitaevskii Equation for Rotating Bose Gases, \emph{Commun. Math. Phys.} \textbf{264} (2006), 505--537.

\bibitem[LSSY]{LSSY} \textsc{E.H. Lieb, R. Seiringer, J.P. Solovej, J. Yngvason}, \textit{The Mathematics of the Bose Gas and its Condensation}, Oberwolfach Seminar Series \textbf{34}, Birkh\"{a}user, Basel (2005), expanded version available at \textit{arXiv:cond-mat/0610117}.







\bibitem[LSY1]{LSY1} 	\textsc{E.H. Lieb, R. Seiringer, J. Yngvason}, Bosons in a Trap: A Rigorous Derivation of the Gross-Pitaevskii Energy Functional, \emph{Phys. Rev. A.} \textbf{61} (2000), 043602.

\bibitem[LSY2]{LSY2} \textsc{E.H. Lieb, R. Seiringer, J. Yngvason}, The Yrast Line of a Rapidly Rotating Bose Gas: The Gross-Pitaevskii Regime, \textit{Phys. Rev. A} \textbf{79} (2009), 063626.

\bibitem[MCWD]{MCWD}	\textsc{K.W. Madison, F. Chevy, W. Wohlleben, J. Dalibard}, Vortex Formation in a Stirred Bose-Einstein Condensate, {\it Phys. Rev. Lett.} {\bf 84}, 806--809 (2000).






\bibitem[RAVXK]{RAVXK}	\textsc{C. Raman, J.R. Abo-Shaeer, J.M. Vogels, K. Xu, W. Ketterle}, Vortex Nucleation in a Stirred Bose-Einstein Condensate, {\it Phys. Rev. Lett.} {\bf 87}, 210402 (2001).

\bibitem[R1]{R1} 		\textsc{N. Rougerie}, The Giant Vortex State for a Bose-Einstein Condensate in a Rotating Anharmonic Trap: Extreme Rotation Regimes, {\it J. Math. Pures Appl.} {\bf 95} (2011), 296--347. 

\bibitem[R2]{R2}		\textsc{N. Rougerie}, Vortex Rings in Fast Rotating Bose-Einstein Condensates, \textit{Arch. Rational Mech. Anal.} \textbf{203}, 69--135 (2012).


















\end{thebibliography}
\end{document}